\documentclass[11pt]{article}
\usepackage[T1]{fontenc}
\usepackage{lmodern,amsmath,amssymb,amsthm,mathtools,mathrsfs,booktabs,microtype}
\usepackage{enumitem}
\usepackage{xcolor}
\usepackage[colorlinks=true,linkcolor=blue!45!black,citecolor=blue!45!black,
  urlcolor=blue!55!black]{hyperref}
\usepackage[nameinlink,noabbrev]{cleveref}
\usepackage[margin=1in]{geometry}

\setlist{itemsep=2pt,topsep=4pt}
\allowdisplaybreaks

\newtheorem{theorem}{Theorem}[section]
\newtheorem{proposition}[theorem]{Proposition}
\newtheorem{lemma}[theorem]{Lemma}
\newtheorem{corollary}[theorem]{Corollary}
\theoremstyle{definition}
\newtheorem{definition}[theorem]{Definition}
\newtheorem{example}[theorem]{Example}
\newtheorem{remark}[theorem]{Remark}
\DeclareMathOperator{\GR}{GR}
\DeclareMathOperator{\ord}{ord}
\DeclareMathOperator{\rad}{rad}

\DeclareMathOperator{\spn}{sp}
\DeclareMathOperator{\Tor}{Tor}
\DeclareMathOperator{\Hull}{Hull}
\DeclareMathOperator{\wt}{wt_H}
\DeclareMathOperator{\remn}{rem}
\DeclareMathOperator{\Tr}{Tr}
\DeclareMathOperator{\rank}{rank}
\newcommand{\F}{\mathbb F}
\newcommand{\Z}{\mathbb Z}
\newcommand{\C}{\mathcal C}
\newcommand{\D}{\mathcal D}
\newcommand{\cc}[3]{c_{#1/#2}(#3)}

\newcommand{\bAS}{b_{\mathrm{AS}}}

\numberwithin{equation}{section}

\begin{document}
\bibliographystyle{abbrv}

\title{Support periods of polynomials and support-primitive decomposition of constacyclic codes over Galois rings}
\author{Li Zhu$^{1}$ and Hongfeng Wu$^{2}$\thanks{Corresponding author.
E-mail addresses: lizhumath@pku.edu.cn (L. Zhu),
whfmath@gmail.com (H. Wu).}
\\[0.5ex]
\small $^{1}$School of Mathematical Sciences, Guizhou Normal University, Guiyang, China
\\
\small $^{2}$College of Science, North China University of Technology, Beijing, China}

\date{}
\maketitle

\thispagestyle{plain}
\setcounter{page}{1}

\begin{abstract}
Let \(A=\GR(p^e,r)\) have residue field \(\F_q\), and let
\(\lambda\in A^\times\). We relate coefficient compression of
monic divisors of \(X^N-\lambda\) to support-primitive decomposition
of constacyclic codes. At a simple-root length, the support period
of every nonconstant monic divisor is preserved at each coefficient
precision and equals the order of the multiplicative stabilizer of
its root set. For a cyclic Frobenius-orbit polynomial, the multiple
equal-difference structure determines its support period, core degree,
and least relative degree of an unramified extension over which
all its basic-irreducible factors are binomials.

Coordinate projections give every constacyclic code over a finite
commutative ring a unique maximal interleaving into copies of one
support-primitive code. At simple-root lengths over \(A\), the
intrinsic period of a code is the greatest common divisor of the
periods of its torsion codes. The decomposition yields distance and
duality formulas and the arithmetic Singleton bound for free codes.
At every length, a nonzero code over \(A\) has the same minimum
Hamming distance as its highest torsion code, which gives a refined
bound. At repeated-root lengths,
identical torsion filtrations can occur with different intrinsic
periods, while a class of nonzero codes obtained by substitution from
simple-root generators retains the same core and minimum distance.

Trace matrices recover the coefficients of basic-irreducible cores
using unit pivots, and local Hensel corrections lift compressed
factors directly. For \(X^n-1\) with \(p\nmid n\), lifting and
output have quasi-linear bit cost in the compressed representation
when the arithmetic data and residue factorizations are given.


\end{abstract}

\section{Introduction}\label{sec:introduction}

For a monic polynomial with nonzero constant coefficient over a
field, a substitution \(f(X)=h(X^s)\) with \(s\) prime to the
characteristic has two equivalent descriptions: its positive
support exponents are divisible by \(s\), and its roots, counted
with multiplicities, are invariant under multiplication by the
\(s\)-th roots of unity. When \(f\) is the monic divisor generating a
constacyclic code,
the same substitution gives an interleaving of \(s\) identical
shorter codes. Over a Galois ring, these descriptions involve
Hensel lifting, the torsion filtration, and the distinction between
simple and repeated roots.

Let \(A=\GR(p^e,r)\), with residue field \(\F_q\), where \(q=p^r\).
We consider monic divisors of \(X^N-\lambda\) and ideals of
\(A[X]/(X^N-\lambda)\), for \(\lambda\in A^\times\).
The support period of a nonconstant polynomial is the greatest
common divisor of its positive exponents with nonzero coefficients.
If this period is \(s\), the unique expression \(f(X)=f^\sharp(X^s)\)
defines its support-primitive core, of degree \(\deg f/s\).

The multiple equal-difference (MED) structure of cyclotomic cosets
describes their partitions into equally spaced subsets \cite{MED}.
For constacyclic codes over finite fields,
support-primitive decomposition gives a canonical shorter code
with the same minimum distance \cite{SP}. The arithmetic Singleton
bound expresses a refinement of the Singleton bound in terms of
the MED structure of the defining set \cite{AS}. The results below
connect these constructions through polynomial and code support
periods over \(A\).

For every nonconstant monic divisor \(f\) of \(X^N-\lambda\) over
a commutative ring, its support period \(s\) divides \(N\), and
its core divides \(Y^{N/s}-\lambda\). When \(p\nmid N\), Hensel
uniqueness preserves \(s\) at every coefficient precision, although
individual coefficients may become nonzero during lifting
(Theorems~\ref{coef:compression} and~\ref{coef:precision}).
The inverse \((N\lambda)^{-1}Xf'\) of the complementary factor
modulo \(f\) gives a local correction that lifts the core directly.
For a basic-irreducible core, a trace matrix with unit determinant
recovers its coefficients even when Newton's identities involve
nonunit multipliers.

At a simple-root length, the support period also equals the order
of the multiplicative stabilizer of the root set. Normalization
of the principal-unit part of \(\lambda\) reduces this description
to Teichm\"uller roots. For a cyclic Frobenius-orbit polynomial,
the MED index determines the core degree and the least relative
degree of an unramified extension over which all its
basic-irreducible factors are binomials. For primitive-order layers,
the identity \(\Phi_d(X)=\Phi_D(X^\mu)\) gives a factorwise
correspondence; layers with the same reduced order \(D\) share
their core factors.

For codes, stability under projections onto coordinate classes
modulo \(s\) defines an intrinsic support period. The admissible
divisors of \(N\) are closed under least common multiples. Hence
every constacyclic code over a finite commutative ring has a
unique maximal interleaving into copies of one support-primitive
code, for arbitrary lengths and without a freeness assumption
(Theorem~\ref{code:intrinsic}). For a monic-divisor code, this
period agrees with the polynomial support period, with an
ambient-length convention for the full code. For an arbitrary
simple-root ideal over \(A\), it is the greatest common divisor
of the periods of all torsion codes
(Theorem~\ref{code:torsion-period}). Equivalently, it is the
stabilizer order of the function assigning to each root the
exponent of its Chinese remainder component.

The interleaving gives formulas for minimum distance, weight
enumerators, covering radii, duals, and hulls. The Singleton bound
for the core yields the arithmetic Singleton bound for free codes.
The highest torsion code, which determines minimum distance, can
have a larger period than the full code and therefore give a
stronger bound (Theorem~\ref{app:singleton}).

At repeated-root lengths, a nonzero code generated by a \(p^a\)-power
substitution from a simple-root generator has the same support-primitive
core, minimum distance, and arithmetic Singleton value as the original
code. The intrinsic
decomposition also applies to arbitrary ideals, but residue
multiplicities and even the full torsion chain need not determine
their intrinsic periods. Explicit examples show how these periods
depend on the lifting coefficients.

At lengths dividing \(q+1\), quadratic trace factors determine hulls
explicitly. Pairing factors with opposite traces at length \(p+1\)
gives exact support periods and minimum distances.

Section~\ref{sec:preliminaries} reviews Galois rings, constacyclic
codes, and MED representations. Sections~\ref{sec:coefficients}
and~\ref{sec:roots} establish the coefficient and root descriptions.
Section~\ref{sec:codes} develops the intrinsic decomposition and
its torsion interpretation, and Section~\ref{sec:applications}
derives the coding-theoretic consequences.
Section~\ref{sec:examples} gives explicit families, and
Appendix~\ref{sec:complexity} gives a certificate for compressed
factorizations and analyzes simultaneous lifting and output costs.

\section{Preliminaries}\label{sec:preliminaries}

Throughout the paper, \(p\) is a prime, \(e,r\) are positive integers,
\(q=p^r\), and
\[
 A=\GR(p^e,r).
\]
Reduction modulo \(p\) is denoted by a bar.  The symbols \(N\) and
\(\lambda\) denote a positive code length and an element of
\(A^\times\), respectively.  A simple-root length means that \(p\nmid N\).
The notation \(\ord_n(q)\) is used only when \(\gcd(n,q)=1\), with
\(\ord_1(q)=1\), and \(\rad(n)\) is the product of the distinct prime
divisors of \(n\), with \(\rad(1)=1\).
We use the following standard facts on Galois rings and constacyclic
codes, together with the ED/MED results in \cite{MED,AS,SP}.

\subsection{Galois rings and monic polynomials}

The ring \(A\) is the finite unramified extension of \(\Z/p^e\Z\) of
degree \(r\).  It is a local principal ideal ring, its maximal ideal is
\(pA\), its residue field is \(\F_q\), and its ideals are exactly
\[
 A\supset pA\supset\cdots\supset p^{e-1}A\supset p^eA=0.
\]
In particular, \(a\in A\) is a unit if and only if \(\bar a\ne0\).
Every nonzero element has the form \(p^j u\), where \(0\le j<e\) and
\(u\in A^\times\).  The integer \(j\) is uniquely determined, although
the unit \(u\) need not be.  We use the convention \(v_p(0)=e\) in
\(A\).  For \(1\le s\le e\), put
\[
 A_s=A/p^sA\simeq\GR(p^s,r).
\]
These descriptions, and the extension theory used here, may be found in
\cite{Wan,McDonald}.

A polynomial is called basic irreducible if it is monic and its reduction
is irreducible of the same degree.  If \(\Pi(Z)\in(\Z/p^e\Z)[Z]\) is
basic irreducible of degree \(r\), then
\[
 A\simeq(\Z/p^e\Z)[Z]/(\Pi(Z)).
\]
Thus \(1,Z,\ldots,Z^{r-1}\) is a free basis over \(\Z/p^e\Z\).
For \(s,t\ge1\) with \(s+t\le e\), multiplication by \(p^s\)
identifies \(A_t\) with \(p^sA_{s+t}\) as \(A\)-modules. Its inverse
defines exact coefficientwise division by \(p^s\) on this ideal.

The Teichm\"uller set of \(A\) is
\[
 \mathcal T_A=\{a\in A:a^q=a\}.
\]
Reduction identifies \(\mathcal T_A\) bijectively with \(\F_q\), and
\(\mathcal T_A^\times\) is cyclic of order \(q-1\).
Every element of \(A\) has a unique expansion
\[
 a=a_0+pa_1+\cdots+p^{e-1}a_{e-1},\qquad a_i\in\mathcal T_A.
\]
The unit group decomposes as
\[
 A^\times=\mathcal T_A^\times\times(1+pA).
\]
Here \(1+pA\) is a finite abelian \(p\)-group.  Consequently, raising to
an integer power prime to \(p\) is an automorphism of this group.
It follows that the principal-unit part of a constacyclic multiplier
has a unique \(N\)-th root whenever \(p\nmid N\).

For every positive integer \(m\), an unramified extension of \(A\) of
degree \(m\) is isomorphic to \(B=\GR(p^e,rm)\).
Its \(A\)-automorphism group is cyclic, generated by the relative
Frobenius automorphism
\[
 \sigma\left(\sum_{j=0}^{e-1}p^j b_j\right)
   =\sum_{j=0}^{e-1}p^j b_j^q,\qquad b_j\in\mathcal T_B.
\]
The fixed ring of \(\sigma\) is \(A\), and \(\sigma(b)=b^q\) for
Teichm\"uller elements. On general elements, \(\sigma\) is given by
the displayed expansion rather than the \(q\)-th power map.

For \(z\in B\), the relative trace and norm are
\[
 \Tr_{B/A}(z)=\sum_{i=0}^{m-1}\sigma^i(z),\qquad
 \operatorname{N}_{B/A}(z)=\prod_{i=0}^{m-1}\sigma^i(z).
\]
Both lie in \(A\). The coefficients of an orbit polynomial are
elementary symmetric functions of the conjugates, valid at every
precision without division. In Newton's identity of index \(j\),
solving for the corresponding elementary symmetric function by
division is valid in \(A\) when \(p\nmid j\).
Proposition~\ref{coef:coefficients} gives the division-free formulas.

Monic division is valid over every commutative coefficient ring.
For a monic polynomial \(G\) of positive degree and an arbitrary
polynomial \(P\), there are unique \(Q,R\) such that
\[
 P=QG+R,\qquad \deg R<\deg G.
\]
We write \(R=\remn_G(P)\).  Multiplication by a monic polynomial is
injective, and the degree of a nonzero product with a monic factor is
the sum of the degrees.  Two monic polynomials over \(A\) are comaximal
if their reductions are coprime.  These elementary facts underlie the
usual coprime Hensel lemma \cite{Wan,McDonald,ModernComputerAlgebra}.

\begin{proposition}\label{pre:hensel}
Let \(P\in A[X]\) be monic, and suppose that
\(\bar P=f_1\cdots f_h\), where the \(f_i\) are pairwise coprime monic
polynomials over \(\F_q\).
There is a unique factorization
\[
 P=F_1\cdots F_h,\qquad
 F_i\text{ monic},\quad \bar F_i=f_i,\quad \deg F_i=\deg f_i.
\]
The factors \(F_i\) are pairwise comaximal.  These factorizations are
compatible with reduction \(A\longrightarrow A_s\).  In particular,
every monic divisor \(f\) of a squarefree \(\bar P\) has a unique monic
lift \(F\mid P\) of the same degree.  If \(f\) is irreducible of degree
\(m\), then \(A[X]/(F)\simeq\GR(p^e,rm)\).
\end{proposition}

When \(p\nmid N\) and \(\lambda\in A^\times\), the polynomial
\(X^N-\bar\lambda\) is squarefree.  Its roots are nonzero and its
derivative is \(NX^{N-1}\).  Proposition~\ref{pre:hensel} therefore
applies to \(X^N-\lambda\), even when \(\lambda\) is not a
Teichm\"uller unit.  In an unramified extension containing the residue
roots, each such root has a unique lift satisfying \(X^N=\lambda\).
Distinct lifted roots have unit difference, and relative Frobenius
permutes them according to its action on their reductions.

Thus monic divisors of \(X^N-\lambda\) correspond bijectively to
monic divisors of its squarefree reduction. This correspondence
also applies to products indexed by any selected union of root orbits.

\subsection{Cyclotomic cosets and orbit factors}

Let \(n\) be positive and coprime to \(p\).
For \(\gamma\in\Z/n\Z\), define the \(q\)-cyclotomic coset
\[
 \cc{n}{q}{\gamma}
   =\{\gamma,\gamma q,\ldots,\gamma q^{m_\gamma-1}\},
 \qquad
 g_\gamma=\gcd(\gamma,n),\qquad d_\gamma=n/g_\gamma.
\]
The gcd may be computed from any integer representative, with
\(\gcd(0,n)=n\). The coset has size
\(m_\gamma=\ord_{d_\gamma}(q)\), and its \emph{leader} is its least
representative in \(\{0,\ldots,n-1\}\).
Its \emph{primitive form} is
\(\cc{d_\gamma}{q}{\gamma/g_\gamma}\), which is independent of the
chosen integer representative. A coset modulo \(n\) is called
\emph{primitive} if its representatives are units modulo \(n\);
the unique coset modulo \(1\) is also regarded as primitive.
The finite-field orbit correspondence is described in
\cite{LidlNiederreiter}, and its Galois-ring version follows from
unramified extension theory and coprime lifting \cite{Wan}.

\begin{proposition}\label{pre:orbit}
Fix a primitive Teichm\"uller \(n\)-th root of unity \(\xi\) in an
unramified extension of \(A\).  The polynomial
\[
 F_{\cc{n}{q}{\gamma}}(X)
   =\prod_{j=0}^{m_\gamma-1}(X-\xi^{\gamma q^j})
\]
belongs to \(A[X]\), is basic irreducible, and has degree \(m_\gamma\).
As the coset varies, these polynomials are exactly the pairwise distinct
basic-irreducible factors of \(X^n-1\).
Over \(\GR(p^e,rt)\), a factor of degree \(m_\gamma\) splits into
\(\gcd(t,m_\gamma)\) basic-irreducible factors, each of degree
\(m_\gamma/\gcd(t,m_\gamma)\); their root sets are the \(q^t\)-orbits
inside \(\cc{n}{q}{\gamma}\).
\end{proposition}

The factor indexed by \(\cc{n}{q}{\gamma}\) has root order
\(d_\gamma\) and degree \(\ord_{d_\gamma}(q)\). Reduction preserves
the prime-to-\(p\) order of a Teichm\"uller root. Specifying \(\xi\)
fixes the labels, and changing \(\xi\) permutes them.
Equal factor degrees can correspond to different root orders;
this distinction enters the seed construction in
Proposition~\ref{root:seed}. Explicit representatives and sizes
are given in \cite{ZhuRepresentatives}.

For a nonconstant squarefree polynomial \(f\in\F_q[X]\) with \(f(0)\ne0\),
its order is the least positive \(n\) such that \(f\mid X^n-1\).
This order is prime to \(p\).  After fixing a primitive \(n\)-th root
\(\zeta\), its defining set is
\[
 T_f=\{\gamma\in\Z/n\Z:f(\zeta^\gamma)=0\}.
\]
It is a union of \(q\)-cyclotomic cosets, and \(|T_f|=\deg f\).
One may also use a larger root-of-unity modulus, provided it is stated
explicitly.  For example, if \(p\nmid N\),
\(\bar\lambda\) has order \(b\), and \(\eta\) is a primitive \(Nb\)-th
root chosen so that \(\eta^N=\bar\lambda\), the roots of
\(X^N-\bar\lambda\) are indexed by
\[
 1+b\Z/Nb\Z.
\]

\subsection{Constacyclic codes, torsion, and duality}

Identify \(A^N\) with
\[
 R_{N,\lambda}=A[X]/(X^N-\lambda)
\]
through coefficient vectors of degree less than \(N\).
An \(A\)-linear code is \(\lambda\)-constacyclic when it is invariant
under
\[
 (a_0,\ldots,a_{N-1})\longmapsto
 (\lambda a_{N-1},a_0,\ldots,a_{N-2}).
\]
Such codes are exactly the ideals of \(R_{N,\lambda}\).
The shift is invertible because \(\lambda\) is a unit.
An element of the quotient always has a unique representative of
degree less than \(N\); support and Hamming weight of codewords refer
to that representative.

If \(F\) is a monic divisor of \(X^N-\lambda\), then
\(\C=(F)\) is free over \(A\), with basis
\[
 F,XF,\ldots,X^{N-\deg F-1}F
\]
when \(\deg F<N\), and with empty basis when \(\deg F=N\).
Thus its rank is \(N-\deg F\).  The constant coefficient of \(F\) is
a unit.  General ideals over a ring with \(e>1\) need not have a monic
divisor as a generator.  The simple-root description below is the
standard Chinese remainder formulation of the structure theory of
codes over finite chain rings \cite{NortonSalageanStructure,Wan}.

\begin{proposition}\label{pre:crt}
Suppose \(p\nmid N\), and write
\[
 X^N-\lambda=\prod_{i=1}^h F_i
\]
for its factorization into distinct basic-irreducible factors.
Put \(d_i=\deg F_i\) and \(A_i=A[X]/(F_i)\).
The Chinese remainder map identifies
\[
 R_{N,\lambda}\simeq\prod_{i=1}^h A_i,\qquad
 A_i\simeq\GR(p^e,rd_i).
\]
Every ideal \(\C\) has a unique description
\[
 \C\simeq\prod_{i=1}^h p^{a_i}A_i,
 \qquad a_i\in\{0,\ldots,e\},
 \qquad
 |\C|=q^{\sum_i(e-a_i)d_i}.
\]
It is free as an \(A\)-module if and only if every \(a_i\) is either
\(0\) or \(e\).  In that case it is generated by the monic divisor
\(\prod_{a_i=e}F_i\).
\end{proposition}

The CRT describes the ideal and its module structure. Hamming weights
are measured in coefficient coordinates, since the CRT map is generally
not a Hamming isometry.

For an arbitrary \(A\)-linear code \(\C\subseteq A^N\), define its
torsion codes by
\[
 (\C:p^j)=\{v\in A^N:p^jv\in\C\},
 \qquad
 \Tor_j(\C)=\overline{(\C:p^j)},\qquad 0\le j<e.
\]
They are codes over \(\F_q\), with
\[
 \bar\C=\Tor_0(\C)\subseteq\Tor_1(\C)\subseteq\cdots
 \subseteq\Tor_{e-1}(\C).
\]
If \(\C\) is constacyclic, each \(\Tor_j(\C)\) is
\(\bar\lambda\)-constacyclic.  We refer to
\cite{NortonSalageanStructure} for this filtration and its module
interpretation.

The elementary-divisor description of an arbitrary linear code is
also useful.  There are uniquely determined nonnegative integers
\(k_0,\ldots,k_{e-1}\) such that, as an \(A\)-module,
\[
 \C\simeq\bigoplus_{a=0}^{e-1}(p^aA)^{k_a}.
\]
An invertible change of ambient coordinates may be used to obtain this
normal form; such a change is not, in general, a coordinate permutation.
The resulting identities for module size and torsion dimensions are
\[
 |\C|=q^{\sum_{a=0}^{e-1}(e-a)k_a},\qquad
 \dim_{\F_q}\Tor_j(\C)=\sum_{a=0}^{j}k_a.
\]
Thus \(\dim\bar\C=k_0\), whereas the last torsion code has dimension
\(\sum_a k_a\), the minimum number of module generators
\cite{NortonSalageanStructure}.

\begin{proposition}\label{pre:torsion}
Under the hypotheses of Proposition~\ref{pre:crt}, the canonical monic
lift \(g_j\in A[X]\) and the generator \(\bar g_j\) of
\(\Tor_j(\C)\) are
\[
 g_j=\prod_{a_i>j}F_i,\qquad
 \bar g_j=\prod_{a_i>j}\bar F_i\quad(0\le j<e).
\]
Here an empty product is \(1\); the generator of the zero residue code
is \(X^N-\bar\lambda\).  Moreover,
\[
 \C=\langle g_0,pg_1,\ldots,p^{e-1}g_{e-1}\rangle,
 \qquad g_{j+1}\mid g_j,
\]
and
\[
 \log_q|\C|=\sum_{j=0}^{e-1}\dim_{\F_q}\Tor_j(\C).
\]
Thus the torsion chain determines the ideal in the simple-root case.
For every nonzero \(A\)-linear code, without a restriction on its
length, one has the standard distance identity
\begin{equation}\label{eq:pre-distance}
 d_H(\C)=d_H(\Tor_{e-1}(\C)).
\end{equation}
\end{proposition}

The distance identity is the socle reduction theorem of Norton and
S\u{a}l\u{a}gean \cite{NortonSalageanDistance}. The map
\[
 \Tor_{e-1}(\C)\longrightarrow \C\cap p^{e-1}A^N,
 \qquad \bar v\longmapsto p^{e-1}v
\]
is a Hamming isometry. A proof of the distance identity and of the
equality \(\Tor_j(\C)=\bar\C\) for free codes is given in
Proposition~\ref{app:highest-torsion}.

The integer \(a_i\) is the depth of the ideal \(p^{a_i}A_i\) in the
\(i\)-th CRT component. The \(j\)-th torsion layer selects the factors
with \(a_i>j\). These depths describe ideals above a squarefree
residue polynomial; root multiplicities instead describe repeated
factors of a polynomial over the residue field.

For an arbitrary code, write
\[
 t(\C)=\dim_{\F_q}\Tor_{e-1}(\C)
\]
for its minimum number of \(A\)-module generators.
We reserve \(\rank_A\C\) for free codes, where it equals \(t(\C)\).
The generator-number and cardinality Singleton bounds are respectively
\begin{equation}\label{eq:pre-singleton}
 d_H(\C)\le N-t(\C)+1,\qquad
 d_H(\C)\le N-\log_{|A|}|\C|+1
 \quad(\C\ne0).
\end{equation}
The second upper bound may be rounded down. These bounds are discussed
in \cite{NortonSalageanDistance}; see \cite{HuffmanPless} for the
field case. A nonzero free code of rank \(k\) is \emph{maximum
distance separable} (MDS) if \(d_H=N-k+1\).

The Euclidean dual is
\[
 \C^\perp=\{v\in A^N:\textstyle\sum_{i=0}^{N-1}v_ic_i=0
                 \text{ for every }c\in\C\}.
\]
Galois rings are finite Frobenius rings; hence
\[
 (\C^\perp)^\perp=\C,\qquad
 |\C|\,|\C^\perp|=|A|^N.
\]
These duality facts are part of the general theory developed in
\cite{Wood}.  The dual of a \(\lambda\)-constacyclic code is
\(\lambda^{-1}\)-constacyclic.  Coordinate permutations preserve the
Euclidean inner product and the Hamming metric.  We write
\(\Hull(\C)=\C\cap\C^\perp\).

For a monic polynomial \(P\) with unit constant coefficient, write
\[
 P^\dagger(X)=P(0)^{-1}X^{\deg P}P(X^{-1})
\]
for its monic reciprocal. Throughout, a monic polynomial is called
self-reciprocal when \(P^\dagger=P\).
If \(X^N-\lambda=FH\) and
\(\C=(F)\), then \(\C^\perp\) is generated by \(H^\dagger\) in
\(R_{N,\lambda^{-1}}\); see \cite{Wan,NortonSalageanStructure}.
More generally the standard torsion-duality relation is
\[
 \Tor_j(\C^\perp)=\Tor_{e-1-j}(\C)^\perp\qquad(0\le j<e),
\]
where the dual on the right is taken over \(\F_q\).
In particular, \(\overline{\C^\perp}=\Tor_{e-1}(\C)^\perp\).
For a free code this becomes \(\overline{\C^\perp}=\bar\C^\perp\).

\subsection{Equal-difference and multiple equal-difference structures}

We recall the MED structure of a cyclotomic coset \cite{MED}
and of a general defining set \cite[Section 3]{AS}.

\begin{definition}\label{pre:eddefinition}
An equal-difference set in \(\Z/n\Z\) is a coset
\[
 E=\gamma+d\Z/n\Z
   =\{\gamma,\gamma+d,\ldots,\gamma+(n/d-1)d\},
 \qquad d\mid n.
\]
The positive divisor \(d\) is its common difference.
A multiple equal-difference representation, abbreviated MED
representation, of a subset \(T\subseteq\Z/n\Z\) is a partition of
\(T\) into equal-difference sets having the same common difference.
One representation is coarser than another if every part of the former
is a union of parts of the latter.
\end{definition}

An ED set with common difference \(d\) has \(n/d\) elements.
For a set of size \(m\), a MED representation with this common
difference has \(md/n\) parts and is uniquely determined whenever
it exists. The finest representation, with \(d=n\), consists of
singletons. Divisibility of common differences reverses the order
given by coarseness of the corresponding partitions.

For a defining set \(T_f\subseteq\Z/n\Z\), put
\[
 \Sigma_f=\{a\in\Z/n\Z:T_f+a=T_f\},\qquad
 \Sigma_f^*=\{d\in\Z_{>0}:d\mid n,\ d\bmod n\in\Sigma_f\}.
\]
Thus \(\Sigma_f\) is a subgroup of the additive cyclic group, while
\(\Sigma_f^*\) is a set of positive integer divisors.  Writing
\[
 d_0=\min\Sigma_f^*
\]
includes the trivial-stabilizer case, where \(d_0=n\).

\begin{proposition}\label{pre:medsets}
Let \(f\in\F_q[X]\) be monic, squarefree, nonconstant, and satisfy
\(f(0)\ne0\).  Use a root-of-unity modulus \(n\) prime to \(p\) with
\(f\mid X^n-1\).
Then \(\Sigma_f=d_0\Z/n\Z\), and
\[
 |\Sigma_f|=\frac n{d_0},\qquad
 \Sigma_f^*=\{d:d_0\mid d,\ d\mid n\}.
\]
For each \(d\in\Sigma_f^*\), the cosets of \(d\Z/n\Z\) contained in
\(T_f\) form its unique MED representation of common difference \(d\).
Every MED representation arises this way.  If \(d\mid d'\), the
representation of common difference \(d\) is coarser than that of
common difference \(d'\).  In particular, \(d_0\) gives the unique
coarsest representation.
\end{proposition}

This is the stabilizer classification in
\cite[Theorem 3.17 and Corollary 3.20]{AS}.
The corresponding root products have a particularly simple form:
\[
 \prod_{j=0}^{n/d-1}(X-\zeta^{\gamma+jd})
       =X^{n/d}-\zeta^{\gamma n/d}.
\]
Thus MED representations give factorizations into binomials of a common
degree over the splitting field.  The binomial belonging to a part
\(\gamma+d\Z/n\Z\) is defined over \(\F_{q^t}\) precisely when
\(d\mid\gamma(q^t-1)\).
If \(n\) is the exact order of \(f\), all binomials in that representation
are defined over \(\F_{q^t}\) precisely when \(d\mid q^t-1\);
see \cite[Section 3.3]{AS}.
Each binomial is irreducible over that field exactly when its ED
part is a single Frobenius orbit.

For a nonconstant polynomial \(f\) with \(f(0)\ne0\), its support period is
\[
 \spn(f)=\gcd\{j\ge1:[X^j]f\ne0\}.
\]
In the squarefree finite-field situation just described, the known
root--coefficient correspondence gives
\begin{equation}\label{eq:pre-support-stabilizer}
 \spn(f)=|\Sigma_f|=\frac n{d_0};
\end{equation}
see \cite[Sections 3--4]{SP}.
Equivalently, for \(s\mid n\), the polynomial belongs to \(\F_q[X^s]\)
if and only if its root set is a union of cosets of the subgroup of
\(s\)-th roots of unity.  For constant divisors and trivial codes, the ambient period in
\eqref{eq:ambient-period} supplies the corresponding convention.

For a single cyclotomic coset the coarsest representation is explicit.
Put
\[
 u_\gamma=\ord_{\rad(d_\gamma)}(q),\qquad
 \omega_\gamma=
 \begin{cases}
  2u_\gamma,&8\mid d_\gamma
             \text{ and }q^{u_\gamma}\equiv3\pmod4,\\
  u_\gamma,&\text{otherwise}.
 \end{cases}
\]
For \(d_\gamma=1\), this gives \(u_\gamma=\omega_\gamma=1\).
The integer \(\omega_\gamma\) is the MED index.

\begin{proposition}\label{pre:med}
The coset \(\cc{n}{q}{\gamma}\) is ED if and only if
\[
 \rad(d_\gamma)\mid q-1,\qquad
 8\mid d_\gamma\ \Longrightarrow\ q\equiv1\pmod4.
\]
Equivalently, it is ED if and only if
\[
 m_\gamma\mid n,\qquad
 \gamma q\equiv\gamma\pmod{n/m_\gamma}.
\]
One has \(\omega_\gamma\mid m_\gamma\), and the unique coarsest MED
representation is
\[
 \cc{n}{q}{\gamma}
   =\bigsqcup_{j=0}^{\omega_\gamma-1}
         \cc{n}{q^{\omega_\gamma}}{\gamma q^j}.
\]
Every part has size
\(\mu_\gamma=m_\gamma/\omega_\gamma\), and their common difference is
\[
 h_\gamma=\frac n{\mu_\gamma}
          =\frac{n\omega_\gamma}{m_\gamma}.
\]
More generally, all parts of the \(q^t\)-orbit decomposition are ED if
and only if \(\omega_\gamma\mid t\).
The leader is
\[
 \min_{0\le j<\omega_\gamma}
        \bigl(\gamma q^j\bmod h_\gamma\bigr),
\]
where residues in the minimum are taken in \(\{0,\ldots,h_\gamma-1\}\).
\end{proposition}

The ED criteria, MED classification, and leader formula appear in
\cite[Sections 3, 4, and 6]{MED}. Here \(\omega_\gamma\) is the number
of coarsest parts, \(\mu_\gamma\) their size, and \(h_\gamma\) their
common difference. The additive stabilizer is
\(h_\gamma\Z/n\Z\), of order \(\mu_\gamma\).
The least extension degree over which the corresponding irreducible
polynomial factors into irreducible binomials is \(\omega_\gamma\)
\cite[Section 5]{MED}.

More precisely, the possible common differences for the coset are
\[
 \left\{\frac{nt}{m_\gamma}:
       \omega_\gamma\mid t,\ t\mid m_\gamma\right\}.
\]
The value indexed by \(t\) corresponds to the decomposition into its
\(q^t\)-orbits, and to the irreducible factorization over
\(\F_{q^t}\subseteq\F_{q^{m_\gamma}}\).
For an arbitrary extension degree, the same decomposition is obtained
by replacing that degree by its greatest common divisor with
\(m_\gamma\).
For a union of several cosets, however, coarsest ED parts can pass
between different cosets.  The stabilizer of the union, rather than
only the individual MED indices, is consequently the datum in
Proposition~\ref{pre:medsets}.

Finally, recall the finite-field arithmetic Singleton bound.
Let \(0\ne\mathcal E=(f)\subsetneq
\F_q[X]/(X^N-\bar\lambda)\), where \(p\nmid N\), and use a
root-of-unity modulus \(n\) for \(f\).  If \(d_0\) is the least common
difference of the defining set, then
\begin{equation}\label{eq:pre-AS}
 d_H(\mathcal E)\le
 \bAS(\mathcal E):=\frac{(\deg f)d_0}{n}+1
                  =\frac{\deg f}{\spn(f)}+1.
\end{equation}
This is the bound of \cite[Section 4]{AS}; its interpretation as the
ordinary Singleton bound of the finite-field support-primitive core
is given in \cite[Section 5]{SP}.
The full-space code has distance one. The ambient-period formulation
in Section~\ref{sec:applications} includes this case.

\section{Coefficient-side description of support periods}\label{sec:coefficients}

\subsection{Compression by the support period}

For a nonconstant polynomial \(f=\sum_{i=0}^d f_iX^i\) over a
commutative ring, put
\[
 \spn(f)=\gcd\{i\geq1:f_i\ne0\}.
\]
The integer \(\spn(f)\) is the \emph{support period} of \(f\)
\cite{SP}. For a nonzero constant, set \(\spn(f)=1\).
A polynomial of support period one is \emph{support-primitive}.
The unique expression
\[
 f(X)=f^\sharp(X^{\spn(f)})
\]
defines its \emph{support-primitive core}. The support includes every
index with a nonzero coefficient, including nonzero nonunits.

\begin{theorem}\label{coef:compression}
Let \(R\) be a commutative ring with identity, let
\(\lambda\in R^\times\), and let \(f\) be a nonconstant monic divisor of
\(X^N-\lambda\). If \(s=\spn(f)\), then
\[
 s\mid N,\qquad
 f(X)=h(X^s),\qquad h\mid Y^{N/s}-\lambda,\qquad \spn(h)=1.
\]
Both \(f(0)\) and \(h(0)\) are units. More generally, for any \(s\mid N\)
and any monic \(h\in R[Y]\),
\[
 h(X^s)\mid X^N-\lambda
 \quad\Longleftrightarrow\quad
 h(Y)\mid Y^{N/s}-\lambda.
\]
\end{theorem}

\begin{proof}
Write \(X^N-\lambda=fQ\). Multiplication by a monic polynomial is
injective in \(R[X]\), and \(Q\) is monic. Evaluation at zero shows that \(f(0)\)
and \(Q(0)\) are units. By definition \(f=h(X^s)\), with \(h\) monic and
of positive degree. Use the direct sum
\[
 R[X]=\bigoplus_{j=0}^{s-1}X^jR[X^s]
\]
to write \(Q=\sum_{j=0}^{s-1}X^jQ_j(X^s)\). If \(N\not\equiv0\pmod s\),
comparison of the component of residue zero gives
\(hQ_0=-\lambda\), contradicting \(\deg h>0\). Thus \(s\mid N\).
Comparing the other components
and using injectivity of multiplication by \(h\) gives \(Q_j=0\)
for \(j\ne0\), and then
\(hQ_0=Y^{N/s}-\lambda\).
The gcd of the positive exponents of \(h\) is one. The same direct-sum
argument proves the more general equivalence; its reverse implication
also follows immediately by substitution.
\end{proof}

For a fixed ambient length \(N\), define
\begin{equation}\label{eq:ambient-period}
 \sigma_N(f)=\gcd\bigl(\{N\}\cup\{i\geq1:f_i\ne0\}\bigr).
\end{equation}
For nonconstant monic divisors, \(\sigma_N(f)=\spn(f)\), whereas
\(\sigma_N(1)=N\). This ambient period includes the constant divisor
in the correspondence with defining-set stabilizers and code periods.

\begin{proposition}\label{coef:grading}
Let \(R\) be a commutative ring and \(s\geq1\). For \(f\in R[X]\),
the following equivalence holds in \(R[X,T]/(T^s-1)\):
\[
 f(TX)=f(X)\quad\Longleftrightarrow\quad f\in R[X^s].
\]
If \(s\mid N\) and \(\lambda\in R\), the map
\[
 \delta_s:R[X]/(X^N-\lambda)\longrightarrow
 R[X]/(X^N-\lambda)\otimes_R R[T]/(T^s-1),\qquad
 X\longmapsto X\otimes T
\]
is a well-defined algebra map. For an ideal \(\mathcal I\), the
inclusion \(\delta_s(\mathcal I)\subseteq
\mathcal I\otimes_R R[T]/(T^s-1)\) is equivalent to stability under
projection onto each exponent class modulo \(s\).
\end{proposition}

\begin{proof}
The elements \(1,T,\ldots,T^{s-1}\) form a free \(R\)-basis. Comparing
their coefficients, and then those of \(X\), proves the first assertion.
The relation \(X^N=\lambda\) is preserved since \(T^N=1\).
Writing a word as \(\sum_{j=0}^{s-1}v_j\), where \(v_j\) contains its
coordinates of indices congruent to \(j\), gives
\(\delta_s(v)=\sum_jv_j\otimes T^j\). Freeness of the second factor
proves the last assertion.
\end{proof}

The map \(\delta_s\) is the coaction of the group scheme \(\mu_s\),
expressed through the grading by exponent classes. In the simple-root
case, evaluation at roots of unity recovers the symmetries of
Section~\ref{sec:roots}. When \(p\mid s\), stability under those
evaluations alone need not imply stability under the coordinate
projections; see Example~\ref{root:repeated-counterexample}.

\subsection{Stability at every precision}

Throughout this subsection let \(A=\GR(p^e,r)\), let \(q=p^r\), and let
\(A_j=A/p^jA\), for \(1\leq j\leq e\). A bar denotes reduction modulo
\(p\).

\begin{theorem}\label{coef:precision}
Suppose that \(p\nmid N\), \(\lambda\in A^\times\), and \(f\) is a
nonconstant monic divisor of \(X^N-\lambda\). Then
\[
 \spn(f\bmod p^j)=\spn(\bar f)\qquad(1\leq j\leq e).
\]
If \(s=\spn(\bar f)\) and \(\bar f=\bar h(X^s)\), then
\[
 f=h(X^s),
\]
where \(h\) is the unique monic Hensel lift of \(\bar h\) which divides
\(Y^{N/s}-\lambda\). In particular, the support-primitive core commutes
with reduction and with Hensel lifting.
\end{theorem}

\begin{proof}
Theorem~\ref{coef:compression} over \(\F_q\) gives \(s\mid N\) and
\(\bar h\mid Y^{N/s}-\bar\lambda\). The latter polynomial is squarefree,
since \(p\nmid N/s\). Lift the coprime factorization containing
\(\bar h\) to obtain a monic divisor \(h\) of \(Y^{N/s}-\lambda\).
Then \(h(X^s)\) divides \(X^N-\lambda\) and reduces to \(\bar f\).
Uniqueness in Hensel lifting for the squarefree polynomial
\(X^N-\bar\lambda\) gives \(f=h(X^s)\).
All exponents of \(f\) are therefore divisible by \(s\), while
reduction shows that \(\spn(f)\) divides \(\spn(\bar f)=s\).
The two integers are equal. The same argument in \(A_j\) proves the
assertion at every precision.
\end{proof}

Individual coefficient positions may change under lifting, even though
their gcd is preserved.

\begin{example}\label{coef:individual-support}
Over \(\Z/4\Z\), one has
\[
 X^7-1=(X-1)(X^3+2X^2+X+3)(X^3+3X^2+2X+3).
\]
The second factor reduces to \(X^3+X+1\) over \(\F_2\), although its
\(X^2\)-coefficient is nonzero over \(\Z/4\Z\). Both polynomials have
support period one.
\end{example}

\begin{proposition}\label{coef:normalization}
Suppose that \(p\nmid N\). Write
\(\lambda=\lambda_0u\), where \(\lambda_0\) is a Teichm\"uller unit
and \(u\in1+pA\). There is a unique \(v\in1+pA\) with \(v^N=u\).
The substitution \(X=vY\) identifies the quotients for
\(X^N-\lambda\) and \(Y^N-\lambda_0\). If \(f\) is monic of degree \(d\),
its normalized polynomial is
\[
 f_0(Y)=v^{-d}f(vY).
\]
Normalization preserves the support, the support period, and monicity.
The quotient isomorphism induced by \(X\mapsto vY\) is a diagonal
Hamming isometry on coefficient vectors.
\end{proposition}

\begin{proof}
The finite abelian group \(1+pA\) is a \(p\)-group. Raising to the
\(N\)-th power is an automorphism because \(\gcd(N,p)=1\).
The substitution carries the first relation to \(u(Y^N-\lambda_0)\).
Every coefficient is multiplied by a unit, which proves the remaining
assertions.
\end{proof}

\subsection{Coefficient formulas and trace reconstruction}

Elementary symmetric functions and exterior powers recover
coefficients from roots without division.

\begin{proposition}\label{coef:coefficients}
Let \(f=X^d+a_1X^{d-1}+\cdots+a_d\) be monic over a commutative ring
\(R\), and let \(M\) be multiplication by the class of \(X\) on the
free module \(R[X]/(f)\). Then
\begin{equation}\label{eq:exterior-coefficients}
 a_j=(-1)^j\Tr\bigl(\textstyle\bigwedge^jM\bigr)\qquad(1\leq j\leq d).
\end{equation}
If \(f=\prod_{i=1}^d(X-\alpha_i)\) in an extension ring, then
\[
 a_j=(-1)^j\sum_{1\leq i_1<\cdots<i_j\leq d}
                  \alpha_{i_1}\cdots\alpha_{i_j}.
\]
Writing \(E_j=(-1)^ja_j\), \(E_0=1\), and
\(P_j=\sum_i\alpha_i^j\), Newton's identities take the form
\begin{equation}\label{eq:newton}
 jE_j=\sum_{i=1}^j(-1)^{i-1}E_{j-i}P_i.
\end{equation}
All three identities hold without any assumption on the
characteristic.
\end{proposition}

\begin{proof}
The characteristic polynomial of the companion operator \(M\) is
\(f\). For any matrix over a commutative ring,
the coefficient of \(Z^j\) in \(\det(I+ZM)\) is the sum of its
principal \(j\)-minors, which is
\(\Tr(\bigwedge^jM)\). This proves \eqref{eq:exterior-coefficients}.
Expansion of a product proves the second formula. Newton's identities
are identities over the integral polynomial ring in the variables
\(\alpha_i\); specialization proves \eqref{eq:newton}.
\end{proof}

For a basic-irreducible divisor of \(X^N-\lambda\), with \(p\nmid N\),
the roots in the proposition are the Frobenius conjugates of a root
\(\alpha\) in an unramified extension.
After Teichm\"uller normalization, these are ordinary \(q\)-power
conjugates. In \(A\), equation~\eqref{eq:newton} determines \(E_j\)
by division only when \(p\nmid j\). The product and exterior-power
formulas apply for every \(j\) and can be evaluated directly on the
compressed core.

Power traces also determine all coefficients through a linear system
whose determinant is a unit. This remains valid when individual Newton
identities have nonunit multipliers.

\begin{proposition}[Reconstruction from a trace matrix]
\label{v2:trace-reconstruction}
Let \(B=\GR(p^e,rd)\) be an unramified extension of
\(A=\GR(p^e,r)\), and let \(\beta\in B\) have residue of degree
\(d\) over \(\F_q\). Write its basic-irreducible polynomial as
\[
 H(Y)=Y^d+c_{d-1}Y^{d-1}+\cdots+c_1Y+c_0.
\]
For relative Frobenius \(\sigma\), write
\(\Tr_{B/A}(z)=\sum_{a=0}^{d-1}\sigma^a(z)\).
For \(k\ge0\), put \(t_k=\Tr_{B/A}(\beta^k)\), and define
\[
 \mathcal H_\beta=(t_{i+j})_{0\le i,j<d},\qquad
 \boldsymbol t=(t_d,t_{d+1},\ldots,t_{2d-1})^{\mathsf T}.
\]
Then \(\det\mathcal H_\beta\in A^\times\), and the coefficients of
\(H\) are the unique solution of
\begin{equation}\label{v2:eq-trace-reconstruction}
 \mathcal H_\beta(c_0,c_1,\ldots,c_{d-1})^{\mathsf T}
       =-\boldsymbol t.
\end{equation}
Equivalently, a monic polynomial \(K\in A[Y]\) of degree \(d\)
equals \(H\) if and only if
\[
 \Tr_{B/A}\bigl(\beta^iK(\beta)\bigr)=0
       \qquad(0\le i<d).
\]
These assertions hold at every coefficient precision, and the linear
system can be solved using only unit pivots.
\end{proposition}

\begin{proof}
Let \(\sigma\) be relative Frobenius and put
\(\beta_a=\sigma^a(\beta)\), for \(0\le a<d\). Their residues
are distinct, so every difference \(\beta_b-\beta_a\), with
\(a\ne b\), is a unit. The Vandermonde matrix
\(V=(\beta_a^j)_{0\le a,j<d}\) satisfies
\[
 \mathcal H_\beta=V^{\mathsf T}V,
 \qquad
 \det\mathcal H_\beta=
     \prod_{0\le a<b<d}(\beta_b-\beta_a)^2.
\]
The determinant belongs to \(A\) and is a unit in \(B\), hence is a
unit in \(A\).

The orbit product \(\prod_a(Y-\beta_a)\) is fixed by Frobenius,
vanishes at \(\beta\), and reduces to the irreducible polynomial
of \(\bar\beta\). It is therefore a monic basic-irreducible relation
of degree \(d\). Any two monic degree-\(d\) relations differ by a
polynomial \(K\) of degree below \(d\) with \(K(\beta)=0\).
Taking \(\Tr_{B/A}(\beta^iK(\beta))\), for \(0\le i<d\),
and using the invertibility of \(\mathcal H_\beta\) makes every
coefficient of \(K\) zero. Thus the orbit product is \(H\).
Multiplying \(H(\beta)=0\) by \(\beta^i\) and taking
the trace gives \eqref{v2:eq-trace-reconstruction}. Invertibility of
\(\mathcal H_\beta\) proves uniqueness and the stated criterion
for \(K\). Its reduction modulo every \(p^j\) is still invertible.
Gaussian elimination can therefore choose each pivot with nonzero
residue and continue on the invertible remaining block; every pivot
is a unit.
\end{proof}

\subsection{A local lifting formula}

The binomial ambient polynomial gives a particularly simple inverse
for the complementary factor.

\begin{theorem}\label{coef:hensel}
Let \(\lambda\in A^\times\), \(p\nmid N\), and
\(1\leq t\leq s\) with \(s+t\leq e\).
Suppose that a positive-degree monic polynomial
\(G_s\in A_s[X]\) divides \(f=X^N-\lambda\). Choose any monic lift \(G\) to \(A_{s+t}[X]\),
with \(\deg G=d\), and form
\[
 R=\remn_G(f),\qquad E=R/p^s\in A_t[X].
\]
Here \(R/p^s\) is defined coefficientwise by the inverse of the
canonical \(A\)-module isomorphism
\(A_t\longrightarrow p^sA_{s+t}\), \(a\longmapsto p^sa\).
Put
\begin{equation}\label{eq:local-correction}
 \Delta=\remn_{G\bmod p^t}
      \bigl(E(N\lambda)^{-1}XG'\bigr)\in A_t[X].
\end{equation}
All factors on the right-hand side are reduced modulo \(p^t\).
Then \(G+p^s\Delta\) is the unique monic divisor of \(f\) modulo
\(p^{s+t}\) which reduces to \(G_s\).
\end{theorem}

\begin{proof}
Let \(Q\in A_s[X]\) be the complementary factor modulo \(p^s\).
Differentiating \(f=QG_s\) and reducing modulo \(G_s\) gives
\[
 QG_s'\equiv NX^{N-1}\pmod {G_s}.
\]
Since \(X^N\equiv\lambda\pmod {G_s}\), multiplication by \(X\) yields
\begin{equation}\label{eq:inverse}
 Q^{-1}\equiv (N\lambda)^{-1}XG_s'\pmod {G_s}.
\end{equation}
Reduce this identity modulo \(p^t\).
Monic division by \(G\) gives \(f=\widetilde QG+p^sE\),
where \(p^sE\) is independent of the chosen lift of \(E\).
By \eqref{eq:local-correction},
\(E-\widetilde Q\Delta\) is divisible by \(G\bmod p^t\); write it
as \(UG\) modulo \(p^t\). Since \(t\leq s\),
\[
 f=(\widetilde Q+p^sU)(G+p^s\Delta)\pmod {p^{s+t}}.
\]
Uniqueness follows from \eqref{eq:inverse}, or from the coprime
Hensel theorem.
\end{proof}

Starting at \(s=1\), this construction permits a precision sequence
\(s\mapsto s+\min(s,e-s)\). It also applies directly to the
support-primitive core, with \(N\) replaced by \(N/\spn(G_s)\).

\begin{proposition}\label{coef:enlargement}
Use the hypotheses and notation of Theorem~\ref{coef:hensel}, and let
\(k\ge1\) satisfy \(p\nmid k\). The correction \(\Delta\) is unchanged
if the ambient polynomial \(X^N-\lambda\) is replaced by
\(X^{kN}-\lambda^k\), with the same monic lift \(G\).
\end{proposition}

\begin{proof}
In the quotient by the chosen lift, write \(X^N=\lambda+R\), where
the coefficients of \(R\) are multiples of \(p^s\).
Since \(t\leq s\), the new remainder is
\(k\lambda^{k-1}R\) modulo \(p^{s+t}\).
The complementary-factor inverse in \eqref{eq:inverse} becomes
\((kN\lambda^k)^{-1}XG'\), which is
\((k\lambda^{k-1})^{-1}\) times the old inverse.
The two changes cancel in \eqref{eq:local-correction}.
\end{proof}

Consequently a cyclic factor may be lifted using the exact order of
its roots, and its compressed core may be lifted at the still smaller
order supplied by Theorem~\ref{root:layer}. These changes affect
arithmetic cost but not the resulting factor.

\subsection{Canonical blocks at repeated-root lengths}

Substitution into the simple-root factorization gives canonical local
blocks at repeated-root lengths.

\begin{proposition}\label{coef:repeated-blocks}
Let \(R=\Z/p^e\Z\), let \(N=p^am\), with \(p\nmid m\), and let
\[
 Y^m-1=\prod_C F_C(Y)
\]
be its factorization into monic basic irreducibles. Then
\begin{equation}\label{eq:canonical-blocks}
 X^N-1=\prod_C F_C(X^{p^a})
\end{equation}
is a product of pairwise comaximal polynomials. For each \(C\), the
quotient \(T_C=R[X]/(F_C(X^{p^a}))\) is local, with maximal ideal
and residue field
\[
 \mathfrak m_C=(p,F_C(X))/(F_C(X^{p^a})),\qquad
 T_C/\mathfrak m_C\simeq\F_p[X]/(\overline{F_C}).
\]
The factorization is unique, up to reordering, among factorizations
into monic pairwise comaximal factors whose quotient rings are local.
If \(F_C(Y)=H_C(Y^{s_C})\), where
\(s_C=\spn(F_C)\), then
\[
 F_C(X^{p^a})=H_C(X^{p^as_C}),\qquad
 \spn(F_C(X^{p^a}))=p^as_C.
\]
These assertions remain valid at every coefficient precision.
\end{proposition}

\begin{proof}
Substitution proves the product identity and the support assertion.
Over \(\F_p\),
\(\overline{F_C}(X^{p^a})=\overline{F_C}(X)^{p^a}\).
Distinct factors therefore have coprime reductions, so are comaximal
over \(R\). The reduction
\(T_C/pT_C\simeq\F_p[X]/(\overline{F_C}^{p^a})\) is local,
with maximal ideal generated by \(\overline{F_C}\).
Since \(pT_C\) is nilpotent and hence lies in every maximal ideal,
taking inverse images proves the asserted maximal ideal and residue
field of \(T_C\).

For uniqueness, consider another factorization with the stated
properties. Pairwise coprimality assigns each entire primary factor
\(\overline{F_C}^{p^a}\) to one residue block. A block with local
quotient cannot contain two distinct residue primary factors, since
reduction modulo the nilpotent ideal generated by \(p\) preserves
locality. Thus both factorizations lift the same coprime primary
factorization, and Proposition~\ref{pre:hensel} identifies their
factors up to order. Precision stability follows from
Theorem~\ref{coef:precision} applied before substitution.
\end{proof}

Over \(A=\GR(p^e,r)\), the same substitution gives comaximal local
blocks. Their residue factors satisfy
\[
 \bar F(X^{p^a})=
 \bigl(\bar F^{[-a]}(X)\bigr)^{p^a},
\]
where \(\bar F^{[-a]}\) is obtained by applying the inverse of the
\(a\)-fold \(p\)-power Frobenius to every coefficient. This automorphism
permutes the irreducible factors of \(Y^m-1\), so distinct substituted
blocks still have coprime reductions. Over \(\F_p\), the coefficient
twist is the identity.

In an unramified extension containing the relevant Teichm\"uller roots,
each block is a product of factors \(X^{p^a}-\beta\).
Choose the Teichm\"uller unit \(\theta\) satisfying
\(\theta^{p^a}=\beta\). Then
\begin{equation}\label{eq:teich-block}
 X^{p^a}-\beta=
 \sum_{j=1}^{p^a}\binom{p^a}{j}
          \theta^{p^a-j}(X-\theta)^j.
\end{equation}
For \(1\leq j<p^a\), the coefficient in this expansion has truncated
\(p\)-adic valuation
\(\min\{e,a-v_p(j)\}\). Indeed
\(\binom{p^a}{j}=(p^a/j)\binom{p^a-1}{j-1}\), and the second factor
is prime to \(p\). Formula \eqref{eq:teich-block} describes the
nilpotent thickening of the root. The blocks in
\eqref{eq:canonical-blocks} give the canonical decomposition of the
quotient ring into local rings; their individual divisors require
further data.

\section{Root-side description of support periods}\label{sec:roots}

At a simple-root length, distinct roots of the ambient constacyclic
polynomial have unit differences in a suitable unramified extension.
This property identifies support periods with multiplicative
stabilizers and relates coefficient compression to MED representations
of defining sets. The corresponding finite-field formulation is given
in \cite{SP}.

Throughout the simple-root part of this section, let
\[
 A=\GR(p^e,r),\qquad q=p^r,\qquad
 A_j=A/p^jA,\qquad A^{(t)}=\GR(p^e,rt),
\]
and assume that \(p\nmid N\) and \(\lambda\in A^\times\).  All unramified
extensions which occur together are embedded in a common unramified
extension.  By a root set we mean the set of distinct roots in such an
extension; the roots of \(X^N-\lambda\) are simple.

\subsection{Unit normalization and multiplicative stabilizers}

Proposition~\ref{coef:normalization} gives a unique \(v\in1+pA\) such that
\[
 \lambda=\lambda_0v^N.
\]
Here \(\lambda_0\) is the Teichm\"uller representative of
\(\bar\lambda\).  If \(F\) is a monic divisor of degree \(d\), the
corresponding monic divisor of \(Y^N-\lambda_0\) is
\[
 F_0(Y)=v^{-d}F(vY)
\]
and its roots are obtained from those of \(F\) by multiplication
by \(v^{-1}\).  The coefficient transformation multiplies the
coefficient of index \(i\) by the unit \(v^{i-d}\), at every
precision.

Choose an unramified extension \(B/A\) in which \(X^N-\lambda\)
splits. Such an extension exists because the distinct residue roots
lift uniquely by Hensel's lemma. Write
\[
 \mu_N(B)=\{\zeta\in B^\times:\zeta^N=1\}.
\]
This is a cyclic group of order \(N\), and the full root set of
\(X^N-\lambda\) is a torsor under \(\mu_N(B)\).  Reduction is a bijection
on this root set and on \(\mu_N(B)\).  For distinct elements of either
set the difference is a unit.

\begin{definition}\label{root:stabilizer-definition}
For a monic divisor \(F\) of \(X^N-\lambda\), let \(R_F\) be its root
set in \(B\).  Its multiplicative stabilizer is
\[
 S_F=\{\zeta\in\mu_N(B):\zeta R_F=R_F\}.
\]
The subgroup \(S_F\), and its order, are independent of the choice of
the splitting extension.
\end{definition}

\begin{theorem}\label{root:stabilizer}
Let \(F\) be a monic divisor of positive degree of \(X^N-\lambda\).
Then, for every \(1\le j\le e\),
\begin{equation}\label{root:period-stabilizer}
 \spn(F\bmod p^j)=\spn(\bar F)=\spn(F)=|S_F|.
\end{equation}
In particular, \(s=|S_F|\) divides \(N\).  There is a unique monic
polynomial \(H\in A[Y]\) such that
\begin{equation}\label{root:core-substitution}
 F(X)=H(X^s).
\end{equation}
It divides \(Y^{N/s}-\lambda\), is support-primitive, and satisfies
\begin{equation}\label{root:core-root-product}
 H(Y)=\prod_{\alpha\in R_F/S_F}(Y-\alpha^s),
\end{equation}
where one representative is chosen from each \(S_F\)-orbit.  All
reductions \(H\bmod p^j\) are support-primitive.
\end{theorem}

\begin{proof}
Put \(s_0=\spn(F)\).  By
Theorems~\ref{coef:compression} and~\ref{coef:precision},
the support period at every precision is \(s_0\), this integer
divides \(N\), and \(F(X)=H(X^{s_0})\) for a unique monic
divisor \(H\) of \(Y^{N/s_0}-\lambda\).  All reductions of
\(H\) are support-primitive.

The action of \(S_F\) on
\(R_F\) is free, since every root is a unit.  Set \(s_1=|S_F|\);
then \(s_1\mid\deg F\) and \(s_1\mid N\).  If \(\zeta\) generates
\(S_F\), the monic polynomials
\(\zeta^{-\deg F}F(\zeta X)\) and \(F(X)\) have the same roots.
Their linear-factor products in \(B[X]\) are therefore equal.
Comparison of their unit constant terms gives
\(\zeta^{\deg F}=1\).  Consequently \(F(\zeta X)=F(X)\).
For a coefficient \(a_i\) of \(F\) this gives
\[
 (\zeta^i-1)a_i=0.
\]
Since \(\zeta\) is a prime-to-\(p\) Teichm\"uller root of unity,
\(\zeta^i-1\) is a unit whenever \(s_1\nmid i\).  Thus \(a_i=0\)
at every such index, and \(s_1\mid s_0\).
Conversely, \(F(X)=H(X^{s_0})\) implies that multiplication by
each \(s_0\)-th root of unity preserves \(R_F\).  Since \(s_0\mid N\),
this gives \(\mu_{s_0}(B)\subseteq S_F\), so \(s_0\mid s_1\).
Therefore \(s_0=s_1\).

Each orbit has the form \(\alpha\mu_s(B)\), and
\[
 \prod_{\zeta\in\mu_s(B)}(X-\alpha\zeta)=X^s-\alpha^s.
\]
Multiplying these identities proves
\eqref{root:core-root-product}.  Distinct orbits have distinct
\(s\)-th powers, and
\((\alpha^s)^{N/s}=\lambda\).
\end{proof}

\begin{remark}\label{root:ambient-period}
The ambient period \(\sigma_N(F)\) from \eqref{eq:ambient-period}
satisfies \(\sigma_N(F)=|S_F|\) for every monic divisor, including
\(F=1\). Indeed, the empty root set has stabilizer \(\mu_N(B)\),
and \(\sigma_N(1)=N\). The full root set also has stabilizer
\(\mu_N(B)\), with \(\spn(X^N-\lambda)=N\). These cases correspond
to the full and zero codes in Section~\ref{sec:codes}.
\end{remark}

\begin{remark}\label{root:support-cancellation}
For a nonconstant divisor, the stabilizer determines the greatest
common divisor of its positive support exponents.
Coefficients within the compressed polynomial may still vanish;
Example~\ref{coef:individual-support} shows that these cancellations
can depend on the precision.
\end{remark}

\subsection{Teichm\"uller defining sets and their coarsest MED representation}

Suppose first that \(\lambda\) is a Teichm\"uller unit of order
\(\ell\), so \(\ell\mid q-1\).  Put \(M=N\ell\) and choose a
primitive Teichm\"uller \(M\)-th root \(\zeta\) with
\(\zeta^N=\lambda\). Such a choice is possible because the reduction map
\((\Z/M\Z)^\times\longrightarrow(\Z/\ell\Z)^\times\) is surjective.
The roots of \(X^N-\lambda\) are precisely
\[
 \{\zeta^\gamma:\gamma\in\Z/M\Z,\ \gamma\equiv1\pmod\ell\}.
\]
For a monic divisor \(F\) of positive degree, let \(T_F\) be its
defining set in this fiber.  It is a union of the \(q\)-cyclotomic
cosets \(\cc{M}{q}{\gamma}\), and
\[
 F(X)=\prod_{\gamma\in T_F}(X-\zeta^\gamma).
\]
Define the additive stabilizer
\[
 \Sigma_F=\{a\in\Z/M\Z:T_F+a=T_F\}.
\]
If \(a\in\Sigma_F\), the nonemptiness of \(T_F\) implies
\(a\equiv0\pmod\ell\).  Consequently
\(\zeta^a\in\mu_N(B)\), and exponentiation identifies
\(\Sigma_F\) with \(S_F\).

\begin{proposition}\label{root:medsets}
Write \(\Sigma_F=d\Z/M\Z\), where \(d\) is a positive divisor
of \(M\).  Then \(\ell\mid d\), and the coarsest MED
representation of \(T_F\) is its partition into cosets of
\(d\Z/M\Z\).  If
\[
 s=\frac{M}{d},\qquad
 T_F=\coprod_{i=1}^{\deg F/s}
             \bigl(\gamma_i+d\Z/M\Z\bigr),
\]
then
\begin{equation}\label{root:med-core}
 \spn(F)=s,\qquad
 F(X)=\prod_{i=1}^{\deg F/s}(X^s-\zeta^{s\gamma_i}),
 \qquad
 F^\sharp(Y)=
       \prod_{i=1}^{\deg F/s}(Y-\zeta^{s\gamma_i}).
\end{equation}
The last polynomial belongs to \(A[Y]\), divides
\(Y^{d/\ell}-\lambda\), and remains support-primitive at every
coefficient precision.
\end{proposition}

\begin{proof}
An equal-difference subset with common difference \(a\mid M\)
is a coset of \(a\Z/M\Z\).  A partition of \(T_F\) into such
subsets exists exactly when \(a\Z/M\Z\subseteq\Sigma_F\).
Among these subgroups, \(\Sigma_F\) is the largest one, giving
the unique coarsest partition; this is the subgroup form of
the MED correspondence in Proposition~\ref{pre:medsets}.
The inclusion \(\Sigma_F\subseteq\ell\Z/M\Z\) gives
\(\ell\mid d\).  Each part contributes
\[
 \prod_{j=0}^{s-1}(X-\zeta^{\gamma_i+jd})
      =X^s-\zeta^{s\gamma_i}.
\]
Theorem~\ref{root:stabilizer} now gives all assertions, since
\(|\Sigma_F|=s\) and \(N/s=d/\ell\).  In particular, the
product defining \(F^\sharp\) has coefficients in \(A\);
the individual binomial constants need not belong to \(A\).
\end{proof}

For a general unit \(\lambda=\lambda_0v^N\), apply this proposition
to \(F_0=v^{-\deg F}F(vY)\).  The same defining set and subgroup
describe \(F\), whose roots are \(v\zeta^\gamma\).  If \(s=|\Sigma_F|\)
and \(h=\deg F/s\), its core is
\begin{equation}\label{root:twisted-core}
 F^\sharp(Y)=v^{sh}F_0^\sharp(Y/v^s).
\end{equation}
Hence normalization preserves the defining-set symmetry and the
compression exponent for every unit parameter at a simple-root length.

Translations of a defining set may interchange distinct Frobenius
orbits. Consequently, a product can have a larger support period than
the gcd of the periods of its irreducible factors. Over a ring of odd
residue characteristic, \(X-1\) and
\(X+1\) each have period \(1\), whereas their product \(X^2-1\)
has period \(2\).

\subsection{Single-orbit compression and unramified binomial factorization}

In this subsection, let \(p\nmid n\) and fix a primitive
Teichm\"uller \(n\)-th root \(\xi\).  For
\[
 c=\cc{n}{q}{a},\qquad
 d=\frac{n}{\gcd(a,n)},\qquad
 m=|c|=\ord_d(q),
\]
write
\[
 u=\ord_{\rad(d)}(q),\qquad
 \omega=\omega(d,q)=
 \begin{cases}
  2u,&8\mid d\ \text{and}\ q^u\equiv3\pmod4,\\
  u,&\text{otherwise}.
 \end{cases}
\]
As in the preliminaries, \(\ord_1(q)=1\), so the values for
\(d=1\) are \(m=\omega=1\).  Let
\[
 F_c(X)=\prod_{b\in c}(X-\xi^b).
\]
Proposition~\ref{pre:med} supplies the ED/MED arithmetic used below.

\begin{proposition}\label{root:binomial-extension}
The polynomial \(F_c\) is a binomial if and only if \(c\) is ED.
In that case
\begin{equation}\label{root:ed-binomial}
 F_c(X)=X^m-\xi^{ma}.
\end{equation}
For any \(t\ge1\), put \(g=\gcd(t,m)\) and \(\nu=m/g\).
The basic-irreducible factors of \(F_c\) over \(A^{(t)}\)
are all binomials if and only if \(\omega\mid t\); in that case
\begin{equation}\label{root:extension-product}
 F_c(X)=\prod_{j=0}^{g-1}
             (X^\nu-\xi^{\nu aq^j}).
\end{equation}
The least unramified extension degree for a complete
basic-irreducible binomial factorization of \(X^n-1\) is
\(\omega(n,q)\), and the admissible degrees are exactly its
positive multiples.
\end{proposition}

\begin{proof}
If \(c\) is ED, its root set is
\(\xi^a\mu_m(B)\), whose product is
\eqref{root:ed-binomial}. Conversely, if \(F_c\) is a binomial,
then \(\spn(F_c)=m\). By Theorem~\ref{root:stabilizer}, its \(m\)
roots form a single orbit of \(\mu_m(B)\), so \(c\) is ED.

By Proposition~\ref{pre:orbit}, base change to \(A^{(t)}\)
splits \(c\) into the \(g\) orbits
\(\cc{n}{q^t}{aq^j}\), \(0\le j<g\), each of size \(\nu\).
Apply the first assertion to these orbits.  The MED theorem
says that all are ED exactly when \(\omega\mid t\), and in
that event gives \eqref{root:extension-product}.

The ED criterion for a primitive order
\(d\mid n\), with \(q^t\) in place of \(q\), is
\[
 \rad(d)\mid q^t-1,\qquad
 q^t\equiv1\pmod4\quad\hbox{if }8\mid d.
\]
These conditions for \(d=n\) imply those for every divisor
\(d\), and they are necessary because the primitive order
\(n\) occurs.  The MED theorem identifies the set of such
\(t\) as the positive multiples of \(\omega(n,q)\).
\end{proof}

\begin{theorem}\label{root:orbit-compression}
With the notation above, put
\[
 \mu=\frac{m}{\omega},\qquad n_0=\frac n\mu,\qquad
 D=\frac d\mu.
\]
Then \(\mu\mid d\), and \(\cc{n_0}{q}{a}\) has size
\(\omega\).  For \(\zeta=\xi^\mu\), the polynomial
\begin{equation}\label{root:orbit-core}
 H_c(Y)=\prod_{j=0}^{\omega-1}(Y-\zeta^{aq^j})
\end{equation}
is basic irreducible over \(A\), has root order \(D\), and
satisfies
\[
 F_c(X)=H_c(X^\mu),\qquad
 \deg H_c=\ord_D(q)=\omega.
\]
For every \(1\le j\le e\) and every \(v\ge1\),
\begin{equation}\label{root:orbit-maximality}
 F_c\bmod p^j\in A_j[X^v]\quad\Longleftrightarrow\quad v\mid\mu.
\end{equation}
In particular, the support period of \(F_c\) is \(\mu\)
at every precision.
\end{theorem}

\begin{proof}
The coarsest MED components of \(c\) are the \(\omega\) ED
orbits for \(q^\omega\), each of size \(\mu\).  Their primitive
forms are ED subsets modulo \(d\); hence \(\mu\mid d\).
In \(\Z/n\Z\) the components are
\begin{equation}\label{root:orbit-fibers}
 aq^j+\{0,n_0,\ldots,(\mu-1)n_0\},
 \qquad 0\le j<\omega.
\end{equation}
They are disjoint full fibers of reduction modulo \(n_0\).
Their distinct images form the orbit
\(\cc{n_0}{q}{a}\) of size \(\omega\).
Since \(\zeta\) has order \(n_0\), the orbit-factor
correspondence proves that \(H_c\) is basic irreducible.
Taking products over the fibers proves the substitution
identity.  The element \(\xi^a\) has order \(d\), so
\((\xi^a)^\mu\) has order \(D\).  This gives
\(\deg H_c=\ord_D(q)=\omega\).

Let \(s=|S_{F_c}|\). The fibers
\eqref{root:orbit-fibers} show that \(\mu\mid s\).
The stabilizer acts freely on \(m\) roots, so \(s\mid m\).
Every ratio of two roots is a \(d\)-th root of unity; hence
\(s\mid d\) as well.  On exponents, stabilizer orbits are
the full fibers of reduction modulo \(n/s\).  The image of
the transitive \(q\)-orbit \(c\) is a \(q\)-orbit of size
\(m/s\).  Thus \(q^{m/s}\) fixes each image.  Its orbit
on a point of \(c\) has size exactly \(s\), and is contained
in the corresponding fiber; it therefore equals that fiber.
All \(q^{m/s}\)-components are ED.  Proposition~\ref{pre:med}
gives \(\omega\mid m/s\), or \(s\mid m/\omega=\mu\).
We conclude that \(s=\mu\).
Theorem~\ref{root:stabilizer} now gives the support period
at every precision.  Membership in \(A_j[X^v]\) is
equivalent to \(v\) dividing every positive support exponent,
which proves \eqref{root:orbit-maximality}.
\end{proof}

\begin{corollary}\label{root:coefficient-descent}
Put \(b_j=\xi^{\mu aq^j}\), \(0\le j<\omega\), and let
\(T_k=e_k(b_0,\ldots,b_{\omega-1})\), with \(T_0=1\).
Then \(T_k\in A\) and
\begin{equation}\label{root:descended-coefficients}
 F_c(X)=\sum_{k=0}^{\omega}(-1)^kT_kX^{\mu(\omega-k)}
       =\prod_{j=0}^{\omega-1}
          \sigma^j(X^\mu-b_0),
\end{equation}
where \(\sigma\) is relative Frobenius on \(A^{(\omega)}/A\).
In particular, \(F_c\) has at most \(\omega+1\) nonzero terms.
Its coefficients are elementary symmetric functions in a
degree-\(\omega\) unramified algebra.
\end{corollary}

\begin{proof}
The elements \(b_j\) are precisely the Frobenius orbit of a
root of \(H_c\).  Expanding \eqref{root:orbit-core} proves
the first formula, and its orbit product gives the second.
The coefficients are fixed by Frobenius and hence belong
to \(A\).
\end{proof}

The last product is the polynomial norm of \(X^\mu-b_0\).
It computes the coefficients in degree \(\omega\), as do the
identities in Proposition~\ref{coef:coefficients}. Vanishing symmetric
functions \(T_k\) can further reduce the number of nonzero terms.

\begin{corollary}[Trace reconstruction of the compressed core]
\label{v2:compressed-traces}
Use the notation of Theorem~\ref{root:orbit-compression}, and put
\(\alpha=\xi^a\), \(\beta=\alpha^\mu\), and
\(B_0=A[\beta]\simeq\GR(p^e,r\omega)\).
The coefficients of \(H_c\) are determined by the
\(\omega\times\omega\) trace system of
Proposition~\ref{v2:trace-reconstruction}, with
\[
 t_k=\Tr_{B_0/A}(\beta^k),\qquad 0\le k<2\omega.
\]
If the traces in \(B=A[\alpha]\simeq\GR(p^e,rm)\) are used instead,
then
\begin{equation}\label{v2:eq-compressed-traces}
 \Tr_{B/A}(\alpha^{\mu k})=\mu t_k,
 \qquad \mu\in A^\times.
\end{equation}
Thus these traces also determine the compressed coefficients by
division by the unit \(\mu\).
\end{corollary}

\begin{proof}
Theorem~\ref{root:orbit-compression} gives
\([B_0:A]=\omega\), \([B:A]=m\), and \(m=\mu\omega\).
In the \(m\)-term Frobenius trace of \(\beta^k\), each of its
\(\omega\) conjugates occurs \(\mu\) times, proving
\eqref{v2:eq-compressed-traces}. The same theorem gives
\(\mu\mid d\mid n\), so \(p\nmid\mu\) and \(\mu\) is a unit
of \(A\). The trace matrix of \(\beta\) is invertible by
Proposition~\ref{v2:trace-reconstruction}.
\end{proof}

\begin{proposition}[Reciprocity and compression]
\label{v2:reciprocal-compression}
Let \(F\) be a nonconstant monic polynomial with unit constant
coefficient, and write \(F(X)=H(X^s)\), where
\(s=\spn(F)\). Then
\[
 F^\dagger(X)=H^\dagger(X^s),\qquad
 \spn(F^\dagger)=s.
\]
In particular, \(F\) is self-reciprocal if and only if \(H\) is
self-reciprocal. For a cyclotomic-orbit factor,
\[
 F_c^\dagger=F_{-c},\qquad H_c^\dagger=H_{-c}.
\]
If \(\deg H_c=2\), and \(\beta\) is a Teichm\"uller root of \(H_c\), then
\begin{equation}\label{v2:eq-reciprocal-quadratic-core}
 H_c(Y)=Y^2-TY+U,\qquad
 F_c(X)=X^{2\mu}-TX^\mu+U,
\end{equation}
where \(T=\beta+\beta^q\in A^\times\) and
\(U=\beta^{q+1}\in A^\times\). Hence \(F_c\) has exactly three
nonzero terms at every coefficient precision. It is self-reciprocal
if and only if \(U=1\), equivalently if the order \(D\) of
\(\beta\) divides \(q+1\).
\end{proposition}

\begin{proof}
If \(h=\deg H\), then \(\deg F=sh\) and \(F(0)=H(0)\).
The definition of the monic reciprocal gives
\[
 F^\dagger(X)=H(0)^{-1}X^{sh}H(X^{-s})=H^\dagger(X^s).
\]
For any monic polynomial \(P\) with unit constant coefficient and
degree \(d\), the support of \(P^\dagger\) is the reflection of
the support of \(P\) about \(d/2\). Since both supports contain
\(0\) and \(d\), their positive exponents have the same gcd.
This proves the period assertion and, by injectivity of substitution,
the equivalence of self-reciprocity.

Reciprocity sends every root \(\xi^b\) to \(\xi^{-b}\), giving
\(F_c^\dagger=F_{-c}\). The two orbits have the same compression
exponent, so their compressed factors satisfy
\(H_c^\dagger=H_{-c}\). In degree two, the roots \(\beta\) and
\(\beta^q\) give \eqref{v2:eq-reciprocal-quadratic-core}, with
\(U\) a unit. If \(\bar T=0\), then
\(\overline{F_c}=X^{2\mu}+\bar U\) has support period
\(2\mu\), contrary to Theorem~\ref{root:orbit-compression}. Thus \(T\) is also a unit,
and all three displayed coefficients remain nonzero at every precision.
Finally,
\(H_c^\dagger(Y)=Y^2-U^{-1}TY+U^{-1}\). Since \(T\) is a unit,
comparison with \(H_c\) shows that self-reciprocity is equivalent
to \(U=1\), which is equivalent to \(D\mid q+1\).
\end{proof}

\subsection{Primitive-order layers}

The compression in Theorem~\ref{root:orbit-compression} applies
simultaneously to all basic-irreducible factors of a cyclotomic
polynomial.

\begin{theorem}\label{root:layer}
Let \(d\ge1\) be prime to \(p\), and put
\[
 m_d=\ord_d(q),\qquad
 \omega_d=\omega(d,q),\qquad
 \mu_d=m_d/\omega_d,\qquad
 D_d=d/\mu_d.
\]
For \(d=1\) all four quantities are \(1\).  Then
\begin{equation}\label{root:layer-arithmetic}
 \rad(D_d)=\rad(d),\qquad
 \ord_{D_d}(q)=\omega_d,\qquad
 \varphi(d)=\mu_d\varphi(D_d).
\end{equation}
Reduction modulo \(D_d\) gives a bijection between
primitive \(q\)-cyclotomic cosets modulo \(d\) and
primitive \(q\)-cyclotomic cosets modulo \(D_d\).
Moreover,
\begin{equation}\label{root:cyclotomic-substitution}
 \Phi_d(X)=\Phi_{D_d}(X^{\mu_d})\qquad\text{in }\Z[X].
\end{equation}
If \(\Phi_{D_d}(Y)=\prod_H H(Y)\) is its
basic-irreducible factorization over \(A\), then
\(\Phi_d(X)=\prod_H H(X^{\mu_d})\) is the
basic-irreducible factorization of \(\Phi_d\).
\end{theorem}

\begin{proof}
The case \(d=1\) is immediate.  Suppose \(d>1\).
The ED component of \(\cc{d}{q}{1}\) containing \(1\)
is
\[
 \langle q^{\omega_d}\rangle
   =\{1+jD_d:0\le j<\mu_d\}
   \subseteq(\Z/d\Z)^\times.
\]
If a prime \(l\mid d\) did not divide \(D_d\), it would
divide \(\mu_d\).  Among the indices \(0\le j<l\le\mu_d\)
one would satisfy \(1+jD_d\equiv0\pmod l\), contrary
to this inclusion.  Therefore the radicals agree.
Euler's formula now gives
\(\varphi(d)/\varphi(D_d)=d/D_d=\mu_d\).
The order formula follows from
Theorem~\ref{root:orbit-compression}.

Because the radicals agree, each unit modulo \(D_d\)
has exactly \(\mu_d\) lifts modulo \(d\), all units.
Theorem~\ref{root:orbit-compression} shows that every
primitive orbit modulo \(d\) is the full inverse image
of its projected orbit modulo \(D_d\).  Conversely,
choose any lift of a unit in a primitive orbit modulo
\(D_d\); its orbit modulo \(d\) has that entire inverse
image.  This proves both surjectivity and injectivity
of the asserted orbit correspondence.

To prove the integral identity, let \(z\) be a complex
root of \(\Phi_{D_d}(X^{\mu_d})\), and let \(h\) be
its order.  Then \(z^d=1\) and
\[
 h\mid d,\qquad h/\gcd(h,\mu_d)=D_d.
\]
For every prime \(l\mid d\), the integer \(D_d\) has
positive \(l\)-adic valuation.  Taking valuations in
the latter equality therefore forces
\[
 v_l(h)=v_l(D_d)+v_l(\mu_d)=v_l(d).
\]
Thus \(h=d\).  Conversely, the \(\mu_d\)-th power
of a primitive \(d\)-th root has order \(D_d\).
The two monic polynomials have these same simple roots
and equal degrees, proving
\eqref{root:cyclotomic-substitution}.  Finally, the
orbit bijection and \eqref{root:orbit-core} identify
each \(H(X^{\mu_d})\) with exactly one primitive-order
\(d\) basic-irreducible factor over \(A\).
\end{proof}

The MED fibers select an exponent \(\mu_d\) for which the integral
identity \eqref{root:cyclotomic-substitution} preserves the
basic-irreducible factors. The same core factors can therefore be used
for every layer with a common value of \(D_d\).

\subsection{Compatible lengths, labels, and primitive seeds}

For \(n'=kn\), choosing primitive roots with
\(\xi=(\xi')^k\) identifies the orbit factors at the two lengths
under exponent scaling. A primitive \(n\)-th root also yields
every basic-irreducible factor of \(X^n-1\) by Frobenius orbit
products. For the analogous single-seed construction over
\(\Z/p^e\Z\), see \cite{WangSeed}.

\begin{proposition}\label{root:length-scaling}
Let \(n'=kn\) with \(p\nmid n'\), choose a primitive
Teichm\"uller \(n'\)-th root \(\xi'\), and put
\(\xi=(\xi')^k\).  For \(c=\cc{n}{q}{a}\), let
\(kc\) denote its image under multiplication by \(k\).
Then
\[
 \cc{n'}{q}{ka}=kc,\qquad
 \min(kc)=k\min(c),\qquad
 F_{kc}^{(n')}(X)=F_c^{(n)}(X).
\]
The orbit size, root order, primitive form, MED index,
support period, and compressed polynomial are preserved.
\end{proposition}

\begin{proof}
Multiplication by \(k\) embeds \(\Z/n\Z\) into
\(\Z/n'\Z\) and commutes with multiplication by \(q\).
If \(b\) is the least nonnegative remainder of
\(aq^j\) modulo \(n\), then \(kb\) is the least
nonnegative remainder of \(kaq^j\) modulo \(n'\).
This proves the orbit and minimum assertions.
Moreover,
\[
 \gcd(ka,kn)=k\gcd(a,n),\qquad
 \frac{n'}{\gcd(ka,n')}=\frac n{\gcd(a,n)},
\]
and division by the gcds gives the same primitive
exponent.  The arithmetic invariants are consequently
unchanged.  Compatible roots give
\[
 \prod_{b\in c}(X-(\xi')^{kb})
      =\prod_{b\in c}(X-\xi^b).
\]
Finally, the common stride and injectivity of
\(H(Y)\mapsto H(X^\mu)\) identify the compressed
polynomials.
\end{proof}

Let \(L(n,q)\) be the set of least
nonnegative coset representatives, and let
\(L^\times(d,q)\) be the representatives of primitive
cosets modulo \(d\), with \(L^\times(1,q)=\{0\}\).
The preceding proposition gives the disjoint decomposition
\begin{equation}\label{root:leader-layers}
 L(n,q)=\coprod_{d\mid n}\frac nd\,L^\times(d,q).
\end{equation}
The MED leader formula from Proposition~\ref{pre:med}
takes the form
\[
 \min\cc{n}{q}{a}
   =\min_{0\le j<\omega}
           \bigl(aq^j\bmod(n/\mu)\bigr),
\]
where remainders are least nonnegative.  In a table
at length \(kn\), selecting leaders divisible by \(k\)
recovers the factors at length \(n\).  Selecting leaders
\(b\) with \(\gcd(b,kn)=k\) recovers exactly the
primitive-order \(n\) layer.  If primitive roots are
selected independently, these statements first require
the unit reindexing described below.

\begin{proposition}\label{root:seed}
Let \(n\ge2\), \(p\nmid n\), and \(L=\ord_n(q)\).
Suppose \(G\mid Y^n-1\) is basic irreducible of degree
\(L\), and let
\[
 B=A[Y]/(G),\qquad \beta=Y+(G).
\]
If \(\beta\) has exact order \(n\), then \(B\) is a
splitting ring for \(X^n-1\), and, using \(\beta\)
to index roots, every basic-irreducible factor is
\begin{equation}\label{root:seed-products}
 F_{\cc{n}{q}{a}}(X)=
       \prod_{j=0}^{|\cc{n}{q}{a}|-1}
                      (X-\beta^{aq^j}).
\end{equation}
Exact order is certified by the finite set of tests
\begin{equation}\label{root:seed-order-tests}
 \beta^{n/l}\ne1\qquad\text{for every prime }l\mid n.
\end{equation}
If an isomorphism to a previously chosen splitting
ring sends \(\beta\) to \(\xi^u\), and \(uv\equiv1\pmod n\),
the factors with the original \(\xi\)-labels are
obtained in \eqref{root:seed-products} by replacing
\(a\) with \(va\).
\end{proposition}

\begin{proof}
Basic irreducibility gives \(B\cong\GR(p^e,rL)\).
Since \(\beta^n=1\) and \(n\) is prime to \(p\),
\(\beta\) is a Teichm\"uller unit.  Exact order
makes its powers the full group of \(n\)-th roots
of unity.  Relative Frobenius sends \(\beta\)
to \(\beta^q\), so its orbit products are precisely
the factors in \eqref{root:seed-products}, by
Proposition~\ref{pre:orbit}.  If the order of
\(\beta\) were a proper divisor of \(n\), it would
divide \(n/l\) for some prime \(l\mid n\); this
proves the order certificate.  Under the stated
isomorphism, \(\beta^{vaq^j}\) maps to
\(\xi^{aq^j}\), giving the reindexing.
\end{proof}

A smaller root order can have the same multiplicative order of \(q\),
so the degree of \(G\) alone does not imply that \(\beta\) has order
\(n\). The tests \eqref{root:seed-order-tests} verify the required
order. For a simple-root constacyclic polynomial with Teichm\"uller
parameter of order \(\ell\), a primitive seed of order \(N\ell\)
supplies the roots in the defining-set fiber. Scaling as in
Proposition~\ref{coef:normalization} then covers arbitrary unit
parameters.

If \(\alpha\) is a root of a basic-irreducible factor of degree \(m\),
then \(1,\alpha,\ldots,\alpha^{m-1}\) is an \(A\)-basis of
\(A[\alpha]\). The unique relation
\[
 \alpha^m=-\sum_{i=0}^{m-1}c_i\alpha^i
\]
recovers the factor \(X^m+\sum_{i=0}^{m-1}c_iX^i\).
The change-of-basis matrix to any other supplied \(A\)-basis
is invertible modulo \(p\), and hence over \(A\), so the coordinates
can be computed using unit pivots. Scaling a Teichm\"uller root
by a unit of \(A\) preserves this power-basis property.
Applying this construction to \(\alpha^\mu\) in
Theorem~\ref{root:orbit-compression} recovers the
degree-\(\omega\) core. The trace construction is given in
Corollary~\ref{v2:compressed-traces}.

\subsection{Obstructions at repeated-root lengths}

Over a finite field, the repeated-root formula of
\cite{SP} expresses the support period as
\(p^{v_f}|\Sigma_f|\), where \(v_f\) is the smallest
\(p\)-adic valuation of a positive root multiplicity
and \(\Sigma_f\) stabilizes the multiplicity function.
This formula depends on the characteristic-\(p\) Frobenius identity.
Over a Galois ring, two monic divisors of \(X^N-\lambda\) can have
identical residue multiplicities but different support periods,
although both generate free codes.

\begin{example}\label{root:repeated-counterexample}
Let \(A=\Z/4\Z\), \(N=4\), and \(\lambda=1\).  Set
\[
 F_1=X^2+1,\qquad F_2=X^2+2X+1=(X-1)^2.
\]
Both are monic divisors of \(X^4-1\), since
\[
 F_1(X)(X^2-1)=X^4-1,\qquad
 F_2(X)(X^2+2X+3)=X^4-1
       \quad\text{in }A[X].
\]
They have the same reduction
\(\bar F_1=\bar F_2=(X-1)^2\), but
\[
 \spn(F_1)=2,\qquad \spn(F_2)=1.
\]
The residue multiplicity function has one root of multiplicity \(2\)
and stabilizer of order \(1\), so the finite-field formula yields
\(2\) for both reductions. It therefore does not determine the ring
support period.

The associated free codes have identical torsion codes;
their different projection periods are computed in
Example~\ref{code:repeated-counterexample}.

The same example gives
\[
 F_2(-X)=F_2(X),\qquad F_2\notin A[X^2].
\]
Here \((-1-1)\cdot2=0\) in \(A\), so invariance under \(X\mapsto-X\)
does not force the odd coefficient to vanish.
\end{example}

\begin{remark}\label{root:no-unramified-roots}
When \(p\mid N\), a general unit parameter need not
admit the normalization in
Proposition~\ref{coef:normalization}, nor need its roots
exist in an unramified extension.  For example,
over \(\Z/p^2\Z\) the polynomial
\[
 X^p-(1+p)
\]
has no root in any \(\GR(p^2,t)\).  A root would
reduce to \(1\), because \(x^p=1\) has only that
root in a field of characteristic \(p\).  It
would consequently have the form \(1+pa\), whose
\(p\)-th power is \(1\) modulo \(p^2\), a
contradiction.  In particular, for \(p=2\) the
polynomial \(X^2-3\) has no root in any unramified
extension of \(\Z/4\Z\).
Thus the coefficient and ideal structure of repeated-root blocks
requires information beyond the residue roots and their multiplicities.
\end{remark}

\section{Support-primitive decomposition of constacyclic codes}
\label{sec:codes}

Coordinate projections give an intrinsic formulation of support-primitive
decomposition over a finite commutative ring, including for ideals without
a monic generator.  For simple-root constacyclic codes over a Galois ring,
the resulting period is determined by the full torsion filtration and by
the CRT exponent function on the roots.  This formulation extends the
finite-field decomposition of \cite{SP}.

\subsection{Coordinate projections and an intrinsic period}

In this subsection, let \(A\) be a nonzero finite commutative ring with identity,
let \(\lambda\in A^\times\), and put
\[
 R_{N,\lambda}=A[X]/(X^N-\lambda).
\]
We identify this ring with \(A^N\) by the unique representatives of degree less
than \(N\).  A \(\lambda\)-constacyclic code is an ideal
\(\C\triangleleft R_{N,\lambda}\).  Multiplication by \(X\) is invertible,
because \(X^{-1}=\lambda^{-1}X^{N-1}\).

Fix a positive divisor \(s\) of \(N\), and write \(M=N/s\).  For
\(0\le j<s\), set
\begin{align*}
 V_j^{(s)}
 &=\left\{\sum_{i=0}^{N-1}a_iX^i:
                 a_i=0\ \text{if }i\not\equiv j\pmod s\right\},\\
 \pi_j^{(s)}\left(\sum_{i=0}^{N-1}a_iX^i\right)
 &=\sum_{\substack{0\le i<N\\i\equiv j\pmod s}}a_iX^i.
\end{align*}
Thus \(R_{N,\lambda}=\bigoplus_{j=0}^{s-1}V_j^{(s)}\) as an
\(A\)-module.  The substitution
\[
 \iota_s:R_{M,\lambda}=A[Y]/(Y^M-\lambda)\longrightarrow R_{N,\lambda},
 \qquad Y\longmapsto X^s
\]
is an injective ring homomorphism, and
\begin{equation}
 R_{N,\lambda}
 =\bigoplus_{j=0}^{s-1}X^j\iota_s(R_{M,\lambda})
 \label{code:ambient-splitting}
\end{equation}
is a free module decomposition over \(\iota_s(R_{M,\lambda})\).

\begin{theorem}
\label{code:projection}
Let \(\C\triangleleft R_{N,\lambda}\) and \(s\mid N\).  The following
conditions are equivalent:
\begin{enumerate}
\item \(\pi_j^{(s)}(\C)\subseteq\C\) for every \(0\le j<s\);
\item
 \(\displaystyle \C=\bigoplus_{j=0}^{s-1}(\C\cap V_j^{(s)})\);
\item there is an ideal \(\D\triangleleft R_{N/s,\lambda}\) such that
 \begin{equation}
  \C=\bigoplus_{j=0}^{s-1}X^j\iota_s(\D).
  \label{code:extension-form}
 \end{equation}
\end{enumerate}
When these conditions hold, the ideal \(\D\) is unique and equals
\(\iota_s^{-1}(\C\cap V_0^{(s)})\).  The map
\begin{equation}
 \Phi_s:\D^{\oplus s}\longrightarrow\C,\qquad
 (d_0,\ldots,d_{s-1})\longmapsto
       \sum_{j=0}^{s-1}X^j d_j(X^s)
 \label{code:interleaver}
\end{equation}
is an \(A\)-linear isomorphism obtained by a permutation of coordinates.
In particular, it preserves Hamming weight.
\end{theorem}

\begin{proof}
The first condition implies the second by applying the ambient direct sum
decomposition to every element of \(\C\).  Conversely, the summands in the
second condition are supported on disjoint coordinate sets, so their
components are precisely the images under the corresponding projections.

Suppose that the second condition holds.  The set
\(\D=\iota_s^{-1}(\C\cap V_0^{(s)})\) is an ideal of \(R_{N/s,\lambda}\):
both \(\C\) and \(V_0^{(s)}=\iota_s(R_{N/s,\lambda})\) are stable under
multiplication by the latter subring.  Since \(X\) is invertible and
\(\C\) is an ideal, multiplication by \(X^j\) gives
\[
 X^j(\C\cap V_0^{(s)})=\C\cap V_j^{(s)}.
\]
For the reverse inclusion, multiplication by \(X^{-j}\) sends a word
supported on \(j,j+s,\ldots,j+(M-1)s\) to one supported on
\(0,s,\ldots,(M-1)s\).  This proves the third condition.
The third condition implies the second.  The contraction formula proves
uniqueness.

Finally, using representatives of degree less than \(M\), the map
\(\Phi_s\) sends the coefficient in position \(i\) of the \(j\)-th input
word to position \(si+j\).  These positions are distinct and exhaust
\(\{0,\ldots,N-1\}\), so \(\Phi_s\) is the restriction of an ambient
coordinate permutation.
\end{proof}

\begin{definition}
\label{code:period-definition}
A divisor \(s\mid N\) satisfying the equivalent conditions of
Theorem~\ref{code:projection} is called an admissible period of \(\C\).
Write \(\mathcal S(\C)\) for the set of admissible periods.
\end{definition}

\begin{theorem}
\label{code:intrinsic}
The set \(\mathcal S(\C)\) is exactly the set of positive divisors of a
uniquely determined divisor \(s(\C)\) of \(N\).  For \(s=s(\C)\), the
ideal \(\D\) in Theorem~\ref{code:projection} satisfies \(s(\D)=1\).
It is uniquely determined as an ideal of \(R_{N/s,\lambda}\).
Moreover,
\[
 s(0)=s(R_{N,\lambda})=N.
\]
\end{theorem}

\begin{proof}
The divisor \(1\) is admissible.  If \(t\mid s\mid N\), each projection
modulo \(t\) is a sum of projections modulo \(s\).  Hence \(s\in
\mathcal S(\C)\) implies \(t\in\mathcal S(\C)\).
If \(s,t\in\mathcal S(\C)\) and \(l=\operatorname{lcm}(s,t)\), then
for \(0\le j<l\),
\[
 \pi_j^{(l)}
 =\pi_{j\bmod s}^{(s)}\pi_{j\bmod t}^{(t)}.
\]
Indeed, both sides retain exactly the coordinates satisfying the two
indicated congruences, or equivalently the single congruence modulo \(l\).
Thus \(l\) is admissible.  The least common multiple of all elements of
the finite set \(\mathcal S(\C)\) is consequently its greatest element,
denoted by \(s(\C)\).  Closure under taking divisors proves the first
assertion.

Let \(\D\) be the ideal obtained at this greatest period.  If \(t>1\)
were an admissible period of \(\D\), then \(t\mid N/s\).  Applying
Theorem~\ref{code:projection} twice would express \(\C\) as the
interleaving of \(st\) copies of one ideal.  Explicitly, the new residue
classes have indices \(j+sk\), with \(0\le j<s\) and \(0\le k<t\).
Thus \(st\in\mathcal S(\C)\), contradicting maximality.  Uniqueness
follows from the contraction description in Theorem~\ref{code:projection}.
Finally, both the zero ideal and the whole ring are stable under every
coordinate projection.
\end{proof}

\begin{definition}
\label{code:core-definition}
The integer \(s(\C)\) is the intrinsic support period of \(\C\).
The ideal \(\C_{\mathrm{sp}}\) obtained by contraction at this period
is its support-primitive core.  The coordinate decomposition
\[
 \C\cong \C_{\mathrm{sp}}^{\oplus s(\C)}
\]
is its support-primitive decomposition.
\end{definition}

The intrinsic period refers to the specified constacyclic coordinate
structure.  Both the zero code and the full code have cores of length
\(1\), consistently with Theorem~\ref{app:duality}.

\begin{remark}
\label{code:coaction}
The decomposition \eqref{code:ambient-splitting} makes \(R_{N,\lambda}\)
a ring graded by \(\mathbb Z/s\mathbb Z\).  Thus admissibility means
exactly that \(\C\) is a homogeneous ideal for this grading.
Equivalently, put \(H_s=A[T]/(T^s-1)\), with its usual Hopf algebra
structure \(\Delta(T)=T\otimes T\).  The substitution
\[
 \delta_s:R_{N,\lambda}\longrightarrow R_{N,\lambda}\otimes_A H_s,
 \qquad X\longmapsto X\otimes T
\]
is well defined because \(s\mid N\), and
\[
 \delta_s(c)=\sum_{j=0}^{s-1}\pi_j^{(s)}(c)\otimes T^j.
\]
Since \(1,T,\ldots,T^{s-1}\) form an \(A\)-basis of \(H_s\),
admissibility is equivalent to
\(\delta_s(\C)\subseteq\C\otimes_A H_s\).
The latter tensor product embeds in \(R_{N,\lambda}\otimes_A H_s\)
because \(H_s\) is free over \(A\).
No invertibility assumption on \(s\) in \(A\) is required.  The coaction
records the action of the group scheme \(\mu_s\), including information
that may be lost on passing to ordinary roots of unity.
\end{remark}

\subsection{Monic divisors and free codes}

For a monic divisor \(f\) of \(X^N-\lambda\), the ambient period
defined in \eqref{eq:ambient-period} reduces to
\begin{equation}
 \sigma_N(f)=
 \begin{cases}
 N,&f=1,\\
 \spn(f),&\deg f>0.
 \end{cases}
 \label{code:sigma}
\end{equation}
Thus \(\spn(1)=1\), whereas \(\sigma_N(1)=N\) agrees with the intrinsic
period of the full code.  For the zero-code generator,
\(\spn(X^N-\lambda)=N\).

\begin{proposition}
\label{code:monic-period}
Let \(A\) be a nonzero finite commutative ring with identity,
\(\lambda\in A^\times\),
and \(f\mid X^N-\lambda\) a monic divisor.  Put \(\C=(f)\).
For every \(s\mid N\),
\[
 s\in\mathcal S(\C)\quad\Longleftrightarrow\quad f\in A[X^s].
\]
Consequently,
\begin{equation}
 s(\C)=\sigma_N(f).
 \label{code:monic-period-equality}
\end{equation}
If \(1\le\deg f<N\), write
\(s=\spn(f)\), \(f(X)=F(X^s)\), and \(M=N/s\).
Then
\[
 F\mid Y^M-\lambda,\qquad
 \C_{\mathrm{sp}}=(F)\subseteq R_{M,\lambda},
 \qquad
 \rank_A\C=s\,\rank_A\C_{\mathrm{sp}}=N-\deg f.
\]
\end{proposition}

\begin{proof}
If \(f=F(X^s)\), multiplication of the direct sum
\eqref{code:ambient-splitting} by \(f\) gives
\[
 (f)=\bigoplus_{j=0}^{s-1}X^j\iota_s((F)).
\]
Monic division gives \(F\mid Y^M-\lambda\).  Theorem~\ref{code:projection}
therefore proves admissibility.

For the converse, the cases \(f=1\) and \(f=X^N-\lambda\) are immediate.
Assume \(0<d=\deg f<N\) and that \(s\) is admissible.
The constant coefficient of \(f\) is a unit, since the constant
coefficients in \(f h=X^N-\lambda\) multiply to \(-\lambda\).
Every element of \((f)\) has a unique representative \(a f\) with
\(\deg a<N-d\): divide a multiplier by
\(h=(X^N-\lambda)/f\), and use injectivity of multiplication by the
monic polynomial \(f\).
In particular,
\(\pi_0^{(s)}(f)=a f\) as ordinary polynomials for such an \(a\).
The left side is nonzero and has degree at most \(d\).
Since \(f\) is monic, \(\deg(a f)=\deg a+d\), so \(a\) is constant.
Comparing constant coefficients and using \(f(0)\in A^\times\)
gives \(a=1\).  Thus \(\pi_0^{(s)}(f)=f\), as required.

Theorem~\ref{coef:compression} gives \(s\mid N\) and
\(F\mid Y^{N/s}-\lambda\) at the maximal polynomial support period.
The uniqueness of the core then identifies it with \((F)\).
The monic multiples \(f,Xf,\ldots,X^{N-d-1}f\) form an \(A\)-basis
of \(\C\); applying the same argument to \(F\) proves the rank formulas.
The endpoint conventions complete \eqref{code:monic-period-equality}.
\end{proof}

\begin{theorem}
\label{code:free-characterization}
Let \(A=\GR(p^e,r)\).  A constacyclic ideal
\(\C\triangleleft R_{N,\lambda}\) is free as an \(A\)-module if and
only if it is generated by a monic divisor of \(X^N-\lambda\).
That monic divisor is unique.
\end{theorem}

\begin{proof}
A monic divisor generates a free code by
Proposition~\ref{code:monic-period}.
Conversely, suppose \(\C\) is free of rank \(K\), and choose an \(A\)-basis.
The matrix of its inclusion in \(A^N\) has rank \(K\) modulo \(p\).  Indeed, if a vector \(u\in A^K\) with a unit coordinate
had image in \(pA^N\), then the nonzero vector \(p^{e-1}u\) would have
zero image, contradicting injectivity.  Hence the residue code
\(\overline{\C}\) has dimension \(K\).

The cases \(K=0,N\) give the generators \(X^N-\lambda\) and \(1\),
respectively.  Suppose \(0<K<N\), put \(d=N-K\), and let
\(\bar f\) be the monic generator of the field constacyclic code
\(\overline{\C}\).  Thus \(\deg\bar f=d\).
Lift the \(K\) vectors
\(\bar f,X\bar f,\ldots,X^{K-1}\bar f\) to \(\C\).
The lifts form an \(A\)-basis, since their change-of-basis matrix is
invertible modulo \(p\).  The submatrix in columns \(d,d+1,\ldots,N-1\) is invertible,
because its reduction is triangular with diagonal entries \(1\).
There is therefore a unique element \(f\in\C\) whose coefficient in
degree \(d\) is \(1\) and whose higher coefficients are zero.  Its reduction
is the unique residue-code word with those highest \(K\) coordinates,
namely \(\bar f\).

The \(K\) vectors \(f,Xf,\ldots,X^{K-1}f\) have the same residue
basis and hence form another \(A\)-basis of \(\C\).
Their highest \(K\) columns are triangular with unit diagonal.
It follows that an element of \(\C\) of degree less than \(d\) is zero.
Divide \(X^N-\lambda\) by the monic polynomial \(f\).
The remainder belongs to \(\C\) and has degree less than \(d\), so it
vanishes.  Thus \(f\mid X^N-\lambda\), and the displayed basis proves
\(\C=(f)\).

Finally, any monic divisor generating \(\C\) has degree \(N-K\)
by the basis formula.  If \(f'\) is another such divisor, the unique
multiplier representation gives \(f'=a f\) with \(a\) constant;
monicity forces \(a=1\).
\end{proof}

Proposition~\ref{code:monic-period} and
Theorem~\ref{code:free-characterization} hold for arbitrary lengths.
For \(e>1\), the principal ideal \(pR_{N,\lambda}\) is nonfree and has no
monic element.  Its intrinsic period is \(N\), and its core is \(pA\).

\subsection{Simple-root ideals, torsion codes, and root labels}

Throughout this subsection, let \(A=\GR(p^e,r)\), let \(q=p^r\),
and suppose \(p\nmid N\).  Write the basic-irreducible factorization as
\[
 X^N-\lambda=\prod_{i=1}^h F_i(X),\qquad m_i=\deg F_i.
\]
Proposition~\ref{pre:crt} gives
\begin{equation}
 R_{N,\lambda}\cong\prod_{i=1}^h A_i,\qquad
 A_i=A[X]/(F_i)\cong\GR(p^e,rm_i).
 \label{code:crt}
\end{equation}
Each ideal \(\C\) is therefore specified by a unique tuple
\begin{equation}
 (a_1,\ldots,a_h)\in\{0,1,\ldots,e\}^h,\qquad
 \C\cong\prod_{i=1}^h p^{a_i}A_i.
 \label{code:exponent-tuple}
\end{equation}
Here \(p^eA_i=0\).  The module structure is
\begin{equation}
 |\C|=q^{\sum_i(e-a_i)m_i},\qquad
 \C\cong_A\bigoplus_{a_i<e}
             \bigl(A/(p^{e-a_i})\bigr)^{m_i}.
 \label{code:module-type}
\end{equation}
In particular, the free ideals are precisely those for which
\(a_i\in\{0,e\}\) for every \(i\).

For the torsion codes of Section~\ref{sec:preliminaries}, set
\begin{equation}
 g_j(X)=\prod_{a_i>j}F_i(X),\qquad 0\le j<e,
 \label{code:torsion-generators}
\end{equation}
where an empty product is \(1\).

\begin{lemma}
\label{code:torsion-description}
With the preceding notation,
\begin{equation}
 \Tor_j(\C)=(\bar g_j),\qquad
 \C=\sum_{j=0}^{e-1}p^j(g_j)
       \quad\text{in }R_{N,\lambda}.
 \label{code:torsion-sum}
\end{equation}
\end{lemma}

\begin{proof}
The first equality is the torsion-generator formula in
Proposition~\ref{pre:torsion}.  For the second, the image of \(g_j\) in
\(A_i\) is zero if \(a_i>j\) and is a unit if \(a_i\le j\), by
comaximality of the factors.  The \(i\)-th CRT component of the sum is
therefore
\(\sum_{j\ge a_i}p^jA_i=p^{a_i}A_i\), with the empty sum interpreted as
zero.
\end{proof}

\begin{theorem}
\label{code:torsion-period}
For a simple-root constacyclic ideal over \(A\),
\begin{equation}
 \mathcal S(\C)=\bigcap_{j=0}^{e-1}\mathcal S(\Tor_j(\C)),
 \qquad
 s(\C)=\gcd_{0\le j<e}\sigma_N(g_j)
       =\gcd_{0\le j<e}\sigma_N(\bar g_j).
 \label{code:period-torsion-gcd}
\end{equation}
The endpoint convention in \eqref{code:sigma} is used for both
coefficient rings.
\end{theorem}

\begin{proof}
Suppose \(s\in\mathcal S(\C)\).  If
\(\bar v\in\Tor_j(\C)\), choose a lift with \(p^jv\in\C\).
For every residue class \(k\pmod s\),
\[
 p^j\pi_k^{(s)}(v)=\pi_k^{(s)}(p^jv)\in\C.
\]
Reduction therefore shows that \(\Tor_j(\C)\) is stable under all
the same projections.

Conversely, suppose \(s\) is admissible for every torsion code.
By Proposition~\ref{code:monic-period} over \(\F_q\), each
\(\bar g_j\) belongs to \(\F_q[X^s]\).
The precision-stability theorem
(Theorem~\ref{coef:precision}) implies
\(g_j\in A[X^s]\); for \(g_j=1\) this is immediate.
Each ideal \((g_j)\), and hence each \(p^j(g_j)\), is stable under
the projections modulo \(s\).  Lemma~\ref{code:torsion-description}
then proves that their sum \(\C\) is stable.
This proves the equality of admissible-period sets.  By
Theorem~\ref{code:intrinsic}, their intersection consists of the
divisors of the gcd in \eqref{code:period-torsion-gcd}.
Theorem~\ref{coef:precision} identifies the two gcd expressions.
\end{proof}

The full torsion filtration also has a root-side description.  Let \(Z\)
be the set of roots of
\(X^N-\bar\lambda\) in an algebraic closure of \(\F_q\).  Define
\[
 a_\C:Z\longrightarrow\{0,\ldots,e\},\qquad
 a_\C(\alpha)=a_i\quad\text{if }\bar F_i(\alpha)=0.
\]
This function is constant on Frobenius orbits.  The group
\(\mu_N\) of \(N\)-th roots of unity acts on \(Z\) by multiplication.
Put
\[
 K(\C)=\{\zeta\in\mu_N:
            a_\C(\zeta\alpha)=a_\C(\alpha)\text{ for every }\alpha\in Z\}.
\]

\begin{corollary}
\label{code:colored-stabilizer}
For \(p\nmid N\),
\begin{equation}
 s(\C)=|K(\C)|.
 \label{code:colored-period}
\end{equation}
Thus the support-primitive decomposition of an arbitrary simple-root
constacyclic ideal is determined by the stabilizer of its CRT exponent
function.
\end{corollary}

\begin{proof}
The defining root set of \(\Tor_j(\C)\) is
\(Z_j=\{\alpha\in Z:a_\C(\alpha)>j\}\).
A multiplier preserves the function \(a_\C\) if and only if it
preserves every one of these threshold sets.  Hence
\[
 K(\C)=\bigcap_{j=0}^{e-1}\operatorname{Stab}_{\mu_N}(Z_j).
\]
The root-stabilizer description
(Theorem~\ref{root:stabilizer}) identifies the order of each subgroup
with \(\sigma_N(\bar g_j)\).  For an empty or full threshold set the
stabilizer is all of \(\mu_N\), of order \(N\), consistently with
\eqref{code:sigma}.  Since \(\mu_N\) is cyclic, the order of the
intersection is the greatest common divisor of the subgroup orders.
Now apply Theorem~\ref{code:torsion-period}.
\end{proof}

For free ideals the exponent function has only the values \(0\) and
\(e\), so Corollary~\ref{code:colored-stabilizer} reduces to the
stabilizer of the defining set of the generator.  In the cyclic
irreducible case, Theorem~\ref{root:orbit-compression} computes this
stabilizer through the MED index.  For nonfree ideals, the intermediate values distinguish roots with
different positive CRT exponents.

\subsection{Repeated-root lengths}

At repeated-root lengths, the basic primary components need not be chain
rings, so the classification \eqref{code:exponent-tuple} can fail.
The intrinsic decomposition and the free-code characterization still
apply, but the torsion filtration need not determine the intrinsic
period.

\begin{example}
\label{code:repeated-counterexample}
In \(R=\Z_4[X]/(X^4-1)\), let
\[
 f_1=X^2+1,\qquad f_2=X^2+2X+1,\qquad \C_i=(f_i).
\]
These are the monic divisors in
Example~\ref{root:repeated-counterexample}.  The resulting codes are free
of rank \(2\), and Proposition~\ref{code:monic-period} gives
\[
 s(\C_1)=2,\qquad s(\C_2)=1.
\]
Their residue generators coincide:
\(\bar f_1=\bar f_2=X^2+1\).
All torsion codes of a free code equal its residue code, as proved in
Proposition~\ref{app:highest-torsion}.  Hence \(\C_1\) and \(\C_2\)
have identical torsion filtrations but different intrinsic periods.

The coaction in Remark~\ref{code:coaction} detects the same obstruction.
The \(\Z_4\)-valued points of \(\mu_2\) are \(1\) and \(-1\), and both
fix \(f_2\).  However,
\[
 f_2(TX)-f_2(X)=2X(T-1)\ne0
       \quad\text{in }\Z_4[X,T]/(T^2-1).
\]
The inequality follows from the free basis \(1,T\).
The odd projection \(2X\) is not in \((f_2)\), since a nonzero
element of that monic-divisor code cannot have degree less than \(2\).
Thus invariance under these two scalar substitutions does not imply
admissibility of the period \(2\).
\end{example}

Proposition~\ref{coef:repeated-blocks} constructs the canonical local
blocks by substitution from simple-root factors. More generally,
substituting \(Y=X^{p^a}\) in a monic divisor of \(Y^m-\lambda\), with
\(p\nmid m\) and \(\lambda\in A^\times\), gives the following family
of free constacyclic codes.

\begin{proposition}\label{v2:substitution-codes}
Let \(A=\GR(p^e,r)\), let \(p\nmid m\), let \(a\ge1\), and put
\(N=p^am\).  Let \(\lambda\in A^\times\), and let \(G\) be a monic
divisor of \(Y^m-\lambda\).  Define
\[
 \C_0=(G)\subseteq A[Y]/(Y^m-\lambda),\qquad
 \C=(G(X^{p^a}))\subseteq A[X]/(X^N-\lambda).
\]
Then the coordinate interleaver gives
\begin{equation}\label{v2:eq-substitution-core}
 \C\cong\C_0^{\oplus p^a},\qquad
 s(\C)=p^a s(\C_0),\qquad
 \C_{\mathrm{sp}}=(\C_0)_{\mathrm{sp}}.
\end{equation}
The core equality is taken in
\(A[Z]/(Z^{m/s(\C_0)}-\lambda)\).  Moreover, for every \(0\le j<e\),
\begin{equation}\label{v2:eq-substitution-torsion}
 \Tor_j(\C)\cong\overline{\C_0}^{\oplus p^a},\qquad
 s(\Tor_j(\C))=s(\C).
\end{equation}
These formulas include the zero and full codes.
\end{proposition}

\begin{proof}
Substitution gives \(G(X^{p^a})\mid X^N-\lambda\), and multiplication
of the ambient decomposition by this generator yields
\[
 \C=\bigoplus_{k=0}^{p^a-1}X^k\iota_{p^a}(\C_0).
\]
This proves the first assertion of \eqref{v2:eq-substitution-core}.
If \(\deg G>0\), substitution multiplies every positive support
exponent by \(p^a\), so
\[
 \spn(G(X^{p^a}))=p^a\spn(G).
\]
Proposition~\ref{code:monic-period} gives the period formula.  Writing
\(G(Y)=H(Y^{s(\C_0)})\), with \(H\) support-primitive, then gives
\(G(X^{p^a})=H(X^{p^as(\C_0)})\), and both cores are generated by
\(H\).  For \(G=1\), the period and core assertions follow from
\(s(\C_0)=m\), \(s(\C)=N\), and the length-one full core.

Multiplication by \(p^j\) and reduction commute with the coordinate
interleaver.  Hence
\[
 \Tor_j(\C)\cong\Tor_j(\C_0)^{\oplus p^a}.
\]
The simple-root code \(\C_0\) is free, so its CRT exponents are all
\(0\) or \(e\).  Proposition~\ref{pre:torsion} therefore gives
\(\Tor_j(\C_0)=\overline{\C_0}\) for every \(j\).
The same support calculation over \(\F_q\) gives
\(s(\Tor_j(\C))=p^a s(\overline{\C_0})\).
Theorem~\ref{coef:precision} and
Proposition~\ref{code:monic-period} identify
\(s(\overline{\C_0})=s(\C_0)\), including the endpoint conventions.
This proves \eqref{v2:eq-substitution-torsion}.
\end{proof}

For \(\lambda=1\), products of the canonical blocks in
Proposition~\ref{coef:repeated-blocks} therefore generate codes
satisfying \(s(\C)=s(\Tor_j(\C))\) for every \(0\le j<e\).

\section{Coding-theoretic consequences}
\label{sec:applications}

Throughout this section \(A=\GR(p^e,r)\), \(q=p^r\), and
\(\lambda\in A^\times\).  Unless a simple-root hypothesis is explicitly
stated, the length \(N\) is arbitrary.  Write
\[
 s=s(\C),\qquad M=N/s,\qquad \D=\C_{\mathrm{sp}}.
\]
Theorem~\ref{code:projection} identifies \(\C\) with \(\D^{\oplus s}\)
by a coordinate permutation.  All statements below apply to nonfree
codes unless freeness is assumed explicitly.  Minimum distances refer
to nonzero codes.

\subsection{Weight distribution and covering radius}

For a code \(\mathcal E\subseteq A^\ell\), let \(A_i(\mathcal E)\)
be the number of its words of Hamming weight \(i\), and set
\[
 W_{\mathcal E}(x,y)=
       \sum_{i=0}^{\ell}A_i(\mathcal E)x^{\ell-i}y^i,\qquad
 \rho(\mathcal E)=
       \max_{v\in A^\ell}\min_{c\in\mathcal E}d_H(v,c).
\]
Both quantities are computed in the ambient space \(A^\ell\).

\begin{theorem}
\label{app:weights}
For every constacyclic ideal,
\begin{equation}
 |\C|=|\D|^s,\qquad
 W_\C(x,y)=W_\D(x,y)^s,\qquad
 \rho(\C)=s\rho(\D).
 \label{app:weight-formulas}
\end{equation}
If \(\C\ne0\) and \(d=d_H(\D)\), then
\begin{equation}
 d_H(\C)=d,\qquad A_d(\C)=sA_d(\D).
 \label{app:distance-minwords}
\end{equation}
If \(\D\) is free of rank \(K\), then \(\C\) is free of rank \(sK\).
More generally, each multiplicity in the \(A\)-module type of
\(\C\) is \(s\) times the corresponding multiplicity for \(\D\).
\end{theorem}

\begin{proof}
The cardinality and module-type assertions follow from the direct
product \(\D^{\oplus s}\).
Hamming weight is additive on the coordinate blocks:
\[
 \wt(d_0,\ldots,d_{s-1})=\sum_{j=0}^{s-1}\wt(d_j).
\]
Multiplication of the \(s\) weight enumerators gives the enumerator
identity.
A tuple has minimum nonzero weight precisely when one component is a
minimum-weight word of \(\D\) and all other components are zero.
This proves both formulas in \eqref{app:distance-minwords}.

For an ambient tuple \(v=(v_0,\ldots,v_{s-1})\), the choices of
closest codewords in the different components are independent.
Therefore
\[
 \min_{c\in\D^{\oplus s}}d_H(v,c)
   =\sum_{j=0}^{s-1}\min_{d_j\in\D}d_H(v_j,d_j).
\]
The maximum of this sum is \(s\rho(\D)\), since each component may
be chosen to attain \(\rho(\D)\).  The ambient coordinate
permutation in Theorem~\ref{code:projection} preserves all these
quantities.
\end{proof}

The enumerator, distance, and covering-radius formulas depend only on
the coordinate product and therefore hold over any finite commutative
ring.  Over \(A\), the module-type formula also records each torsion
multiplicity.

\subsection{Euclidean duality and the hull}

For the Euclidean dual defined in Section~\ref{sec:preliminaries}, the
Frobenius-ring identities
\begin{equation}
 |\mathcal E|\,|\mathcal E^\perp|=|A|^\ell,\qquad
 \mathcal E^{\perp\perp}=\mathcal E
 \label{app:frobenius-duality}
\end{equation}
hold for every \(A\)-submodule \(\mathcal E\subseteq A^\ell\), including
nonfree submodules; see \cite{Wood,Wan}.

\begin{theorem}
\label{app:duality}
The Euclidean dual of a \(\lambda\)-constacyclic code is a
\(\lambda^{-1}\)-constacyclic code.  Moreover,
\begin{equation}
 \mathcal S(\C^\perp)=\mathcal S(\C),\qquad
 s(\C^\perp)=s(\C),\qquad
 (\C^\perp)_{\mathrm{sp}}=\D^\perp.
 \label{app:dual-core}
\end{equation}
The last equality is an equality of ideals in
\(R_{M,\lambda^{-1}}\), and the ambient interleaver gives
\begin{equation}
 \C^\perp\cong(\D^\perp)^{\oplus s}.
 \label{app:dual-sum}
\end{equation}
\end{theorem}

\begin{proof}
Let \(\tau_\lambda\) denote the constacyclic shift.  Directly from
the coordinate formula,
\[
 \langle\tau_\lambda(c),\tau_{\lambda^{-1}}(v)\rangle
       =\langle c,v\rangle.
\]
Since \(\tau_\lambda(\C)=\C\), the dual is invariant under
\(\tau_{\lambda^{-1}}\).

For every \(s\mid N\) the projections are self-adjoint:
\(\langle\pi_j^{(s)}v,c\rangle
=\langle v,\pi_j^{(s)}c\rangle\).
If the projections preserve \(\C\), they therefore preserve
\(\C^\perp\).  Applying this implication to \(\C^\perp\) and using
\eqref{app:frobenius-duality} gives the converse.
This proves equality of the admissible-period sets and their maxima.

The coordinate permutation \(\Phi_s\) preserves inner products.  A
tuple \((v_0,\ldots,v_{s-1})\) is orthogonal to \(\D^{\oplus s}\)
if and only if \(v_j\in\D^\perp\) for each \(j\): test against words
with only the \(j\)-th component nonzero.  This proves
\eqref{app:dual-sum}.  Since \(s\) is also the greatest admissible
period of the dual, uniqueness in Theorem~\ref{code:intrinsic}
identifies the core exactly with \(\D^\perp\).
\end{proof}

For a free code \(\C=(f)\), the dual is generated by the monic
reciprocal \(h^\dagger\) of \(h=(X^N-\lambda)/f\), as recalled in
Section~\ref{sec:preliminaries}.  Formula~\eqref{app:dual-core} also
applies to nonfree codes and, with the convention
\(s(0)=s(R_{N,\lambda})=N\), to both endpoint codes.

\begin{corollary}
\label{app:hull}
The Euclidean hull satisfies
\[
 \Hull(\C)\cong\Hull(\D)^{\oplus s}.
\]
Consequently \(\C\) is linear complementary dual (LCD), self-orthogonal,
dual-containing, or self-dual if and only if \(\D\) has the corresponding
property.  Here LCD means that the Euclidean hull is zero.
\end{corollary}

\begin{proof}
The same ambient permutation sends \(\C\) to \(\D^{\oplus s}\)
and \(\C^\perp\) to \((\D^\perp)^{\oplus s}\).  It therefore sends
their intersection to \((\D\cap\D^\perp)^{\oplus s}\).
The stated equivalences follow by checking intersections, inclusions,
and equality componentwise.
\end{proof}

\subsection{Distance and the highest torsion code}

Minimum Hamming distance is determined by the highest torsion code
\cite{NortonSalageanDistance}.  Its support period can exceed the
intrinsic period of the original code.

\begin{proposition}
\label{app:highest-torsion}
Let \(0\ne\mathcal E\subseteq A^\ell\) be an \(A\)-linear code and put
\(\mathcal T=\Tor_{e-1}(\mathcal E)\).  Then
\begin{equation}
 d_H(\mathcal E)=d_H(\mathcal T).
 \label{app:torsion-distance}
\end{equation}
Furthermore, \(t(\mathcal E):=\dim_{\F_q}\mathcal T\) is the minimum
number of \(A\)-module generators of \(\mathcal E\).
If \(\mathcal E\) is free of rank \(K\), then
\[
 \Tor_j(\mathcal E)=\overline{\mathcal E}\quad(0\le j<e),
 \qquad
 t(\mathcal E)=K,\qquad
 d_H(\mathcal E)=d_H(\overline{\mathcal E}).
\]
\end{proposition}

\begin{proof}
The map
\[
 \F_q^\ell\longrightarrow p^{e-1}A^\ell,\qquad
 \bar v\longmapsto p^{e-1}v
\]
is well defined, is \(\F_q\)-linear for the natural residue-field
structure on \(p^{e-1}A^\ell\), and preserves Hamming weight.  It maps
\(\mathcal T\) bijectively onto \(\mathcal E\cap p^{e-1}A^\ell\).
For a nonzero \(c\in\mathcal E\), let \(a<e\) be the least valuation
of its nonzero coordinates.  Then
\(p^{e-1-a}c\) is a nonzero element of this intersection, and its
support is contained in that of \(c\).
Thus the minimum nonzero weight is attained in this intersection,
proving \eqref{app:torsion-distance}.

By the structure theorem for modules over a finite chain ring,
\(\mathcal E\cong\bigoplus_{\nu=1}^t A/(p^{b_\nu})\), where
\(1\le b_\nu\le e\).  Since \(\dim_{\F_q}(\mathcal E/p\mathcal E)=t\), the integer \(t\)
is the minimum number of generators.  Its socle, the submodule
annihilated by \(p\), also has \(\F_q\)-dimension \(t\).
Inside \(A^\ell\), the elements annihilated by \(p\) are precisely
those in \(p^{e-1}A^\ell\).  The preceding identification of the
socle with \(\mathcal T\) proves the assertion about \(t(\mathcal E)\).

If \(\mathcal E\) is free, its basis matrix has independent reductions,
by the argument in Theorem~\ref{code:free-characterization}.
A \(K\times K\) minor is therefore a unit, and projection onto those
positions identifies \(\mathcal E\) with \(A^K\).  In particular
\(\mathcal E\) is a direct summand of \(A^\ell\).  Write
\(A^\ell=\mathcal E\oplus\mathcal F\), with \(\mathcal F\) free.
If \(p^jv\in\mathcal E\) and \(v=c+w\) in this decomposition, then
\(p^jw=0\), so \(w\in p^{e-j}\mathcal F\) for \(j<e\).
Hence \(\bar w=0\) and \(\bar v\in\overline{\mathcal E}\).
The reverse inclusion holds because \(p^j\mathcal E\subseteq\mathcal E\).
All torsion codes therefore equal \(\overline{\mathcal E}\), which has
dimension \(K\); the distance statement follows from
\eqref{app:torsion-distance}.
\end{proof}

For a nonzero constacyclic code define
\begin{equation}
 \mathcal T=\Tor_{e-1}(\C),\qquad
 t=t(\C)=\dim_{\F_q}\mathcal T,\qquad
 s_*=s(\mathcal T).
 \label{app:star-period}
\end{equation}
The field-code period is the intrinsic period, with \(s_*=N\) when
\(\mathcal T=\F_q^N\).  No restriction on \(N\) is needed.

\begin{lemma}
\label{app:period-comparison}
For every nonzero constacyclic code,
\[
 s(\C)\mid s_*,\qquad s_*\mid t.
 \]
The highest torsion code of the intrinsic core satisfies
\[
 \Tor_{e-1}(\C)\cong\Tor_{e-1}(\D)^{\oplus s}.
 \]
If \(\C\) is free and \(p\nmid N\), then \(s_*=s(\C)\).
\end{lemma}

\begin{proof}
Any coordinate projection which preserves \(\C\) preserves every
torsion code, as in the first part of the proof of
Theorem~\ref{code:torsion-period}; that argument does not require
\(p\nmid N\).  Hence \(s\) is an admissible period of \(\mathcal T\),
so \(s\mid s_*\).  Its decomposition into \(s_*\) identical field
cores gives \(s_*\mid\dim_{\F_q}\mathcal T=t\).
The condition \(p^{e-1}v\in\C\) holds if and only if it holds in
each component of the interleaver.  Reducing componentwise proves
the torsion decomposition.

Finally, in the free case Proposition~\ref{app:highest-torsion}
identifies \(\mathcal T\) with the residue code.
At a simple-root length,
Theorem~\ref{coef:precision} and
Proposition~\ref{code:monic-period} give the same period for its
monic generator and its reduction.  The full-code endpoint also
has period \(N\) on both sides.
\end{proof}

\subsection{Singleton bounds for the two cores}

The field Singleton bound gives \(d_H\le L-K+1\) for a nonzero linear
\([L,K]\)-code \cite{HuffmanPless}.  Proposition~\ref{app:highest-torsion}
gives the same bound for a free \(A\)-code of length \(L\) and rank \(K\).

\begin{theorem}
\label{app:singleton}
Let \(0\ne\C\triangleleft R_{N,\lambda}\), with the notation
\eqref{app:star-period}.  Then
\begin{equation}
 d_H(\C)\le\frac{N-t}{s_*}+1
              \le\frac{N-t}{s(\C)}+1.
 \label{app:as-bounds}
\end{equation}
Both upper bounds are integers.
If \(\mathcal T_{\mathrm{sp}}\) is the support-primitive core of
\(\mathcal T=\Tor_{e-1}(\C)\), then
\[
 d_H(\C)=\frac{N-t}{s_*}+1
 \quad\Longleftrightarrow\quad
 \mathcal T_{\mathrm{sp}}\text{ is an MDS code over }\F_q.
 \]
\end{theorem}

\begin{proof}
The field code \(\mathcal T_{\mathrm{sp}}\) has length \(N/s_*\)
and dimension \(t/s_*\).  Its distance equals that of \(\mathcal T\)
by Theorem~\ref{app:weights}, and that distance equals \(d_H(\C)\)
by Proposition~\ref{app:highest-torsion}.  Applying the Singleton
bound to this field core gives the first inequality in
\eqref{app:as-bounds}.
Its equality condition is exactly the stated MDS condition.
Since \(s\mid s_*\), and \(s_*\) divides both \(N\) and \(t\),
both upper bounds are integers.  The second inequality follows from
\(N-t\ge0\).
Equivalently, the rightmost value in \eqref{app:as-bounds} is the
Singleton bound for \(\Tor_{e-1}(\D)\), which has length \(N/s\)
and dimension \(t/s\).
\end{proof}

The first expression in \eqref{app:as-bounds} is the
highest-torsion refinement of the arithmetic Singleton bound.
The second uses the intrinsic core of \(\C\).  For free codes, the
latter extends the finite-field arithmetic Singleton value of
\cite{AS,SP}.

\begin{corollary}
\label{app:free-singleton}
Let \(0\ne\C=(f)\) be free, where \(f\) is the monic divisor of
\(X^N-\lambda\) generating it.  Define
\begin{equation}
 \bAS(\C):=\frac{\deg f}{\sigma_N(f)}+1.
 \label{app:as-definition}
\end{equation}
Then
\begin{equation}
 d_H(\C)\le\bAS(\C)
          =M-\rank_A\C_{\mathrm{sp}}+1.
 \label{app:free-as}
\end{equation}
The code attains this bound if and only if
\(\C_{\mathrm{sp}}\) is an MDS free \(A\)-code; equivalently,
its residue code is MDS over \(\F_q\).
If \(p\nmid N\), then the highest-torsion refinement coincides
with \(\bAS(\C)\).
\end{corollary}

\begin{proof}
The free rank is \(t=N-\deg f\), by
Propositions~\ref{code:monic-period} and
\ref{app:highest-torsion}.
Substituting into \eqref{app:as-bounds} gives
\eqref{app:free-as}.
The core is free of rank \(t/s\) and has the same distance as \(\C\).
Its residue code has dimension \(t/s\) and the same distance.
Consequently equality in \eqref{app:free-as} is precisely the
Singleton equality for the core, or for its residue code.
The last assertion follows from
Lemma~\ref{app:period-comparison}.
\end{proof}

\begin{corollary}\label{v2:substitution-distance}
With the notation of Proposition~\ref{v2:substitution-codes}, suppose
\(\deg G<m\).  Then
\begin{equation}\label{v2:eq-substitution-distance}
 d_H(\C)=d_H(\C_0),\qquad
 \bAS(\C)=\bAS(\C_0),\qquad
 s_*=s(\C),
\end{equation}
where \(s_*=s(\Tor_{e-1}(\C))\).
Consequently the highest-torsion bound equals the arithmetic Singleton
bound for \(\C\).  Equality in this bound holds if and only if
\((\C_0)_{\mathrm{sp}}\) is an MDS free \(A\)-code, equivalently if
its residue code is MDS over \(\F_q\).
\end{corollary}

\begin{proof}
The degree assumption makes both codes nonzero.
Theorem~\ref{app:weights} applied to
\(\C\cong\C_0^{\oplus p^a}\) gives the distance equality.
Since
\[
 \deg G(X^{p^a})=p^a\deg G,\qquad
 \sigma_N(G(X^{p^a}))=p^a\sigma_m(G),
\]
definition \eqref{app:as-definition} gives equality of the arithmetic
Singleton values.  The last equality in
\eqref{v2:eq-substitution-distance} is
\eqref{v2:eq-substitution-torsion} with \(j=e-1\).
The common-core identity in \eqref{v2:eq-substitution-core} and
Corollary~\ref{app:free-singleton} give the stated MDS criterion.
\end{proof}

\begin{example}
\label{app:nonfree-z9}
Let \(A=\Z_9\), \(N=4\), and
\[
 X^4-1=(X-1)(X+1)(X^2+1).
 \]
The quadratic is basic irreducible, since \(X^2+1\) has no root
over \(\F_3\).
Consider the nonfree cyclic ideal
\[
 \C=\bigl((X+1)(X^2+1),\,3(X^2+1)\bigr).
 \]
Its CRT exponent tuple, in the displayed order of factors, is
\((0,1,2)\).  Indeed, the first generator is a unit only in the
\(X-1\) component, and the second is a unit times \(3\) in each of
the two linear components and vanishes in the quadratic component.
Thus \(\C\cong_A A\oplus A/(3)\), \(|\C|=27\), and \(t(\C)=2\).
The torsion generators are
\[
 \bar g_0=X^3+X^2+X+1,\qquad
 \bar g_1=X^2+1.
 \]
Their support periods are \(1\) and \(2\), respectively, so
\[
 s(\C)=1,\qquad s_*=2,\qquad
 \frac{N-t}{s(\C)}+1=3,\qquad\frac{N-t}{s_*}+1=2.
 \]
The words have the explicit description
\[
 \C=\{(a+3b,a,a+3b,a):a\in\Z_9,\ b\in\Z/3\Z\}.
 \]
Indeed, the two generators have coefficient vectors
\((1,1,1,1)\) and \((3,0,3,0)\).  Their span is shift-invariant, since
\((0,3,0,3)=3(1,1,1,1)-(3,0,3,0)\).
The parametrization is injective and has \(27\) words.
Every nonzero word has weight \(2\) or \(4\), and the second
generator has weight \(2\).  Therefore \(d_H(\C)=2\).
The highest-torsion bound is attained, while the intrinsic period
is \(1\).
\end{example}

The same difference can occur for free codes at repeated-root
lengths.  In Example~\ref{code:repeated-counterexample}, the code
\((X^2+2X+1)\) has \(s=1\), \(s_*=2\), and distance \(2\).
Its intrinsic-core value is \(3\), whereas its highest-torsion value
is \(2\).  Thus the simple-root hypothesis in the last assertion of
Corollary~\ref{app:free-singleton} cannot be removed.

\subsection{Products of cyclic factors and MED estimates}

The periods of individual Frobenius orbit factors give upper bounds
for the arithmetic Singleton value of a code generated by their product.
Assume throughout this subsection that \(\lambda=1\) and \(p\nmid N\).
For \(d\mid N\), put
\[
 m_d=\ord_d(q),\qquad
 \omega_d=\omega(d,q),\qquad
 \mu_d=\frac{m_d}{\omega_d},
\]
where \(\omega(d,q)\) is the MED index recalled in
Proposition~\ref{pre:med}, and \(m_1=\omega_1=\mu_1=1\).
The MED compression theorem
(Theorem~\ref{root:orbit-compression}) gives, for each
basic-irreducible factor of root order \(d\),
\[
 \deg F=m_d,\qquad \spn(F)=\mu_d.
 \]
Cancellation of coefficients can make the support period of a product
larger than the gcd of the individual periods.

\begin{theorem}
\label{app:common-order}
Let \(L\ge1\), and let \(F_1,\ldots,F_L\) be distinct
basic-irreducible factors of \(X^N-1\), all of root order \(d\).
Suppose
\[
 f=\prod_{i=1}^L F_i
 \]
generates a nonzero code \(\C=(f)\).  Then
\begin{equation}
 d_H(\C)\le\bAS(\C)\le L\omega_d+1
                    \le2L\ord_{\rad(d)}(q)+1.
 \label{app:common-order-bound}
\end{equation}
\end{theorem}

\begin{proof}
Every \(F_i\) belongs to \(A[X^{\mu_d}]\), so their product also
belongs to this subring.  Hence \(\mu_d\mid\spn(f)\).
Its degree is \(Lm_d\).
Corollary~\ref{app:free-singleton} therefore gives
\[
 \bAS(\C)
 =\frac{Lm_d}{\spn(f)}+1
 \le\frac{Lm_d}{\mu_d}+1
 =L\omega_d+1.
 \]
The last inequality in \eqref{app:common-order-bound} is the
definition of the MED index: it is either
\(\ord_{\rad(d)}(q)\) or twice that integer.
\end{proof}

\begin{proposition}
\label{app:mixed-orders}
Let \(L\ge1\), and let \(F_1,\ldots,F_L\) be distinct
basic-irreducible factors of \(X^N-1\), with root orders \(d_i\).
Suppose their product \(f\) generates a nonzero code.  Put
\[
 \mu_i=\frac{\ord_{d_i}(q)}{\omega(d_i,q)},\qquad
 \delta=\gcd(\mu_1,\ldots,\mu_L).
 \]
Then
\begin{equation}
 \delta\mid\spn(f),\qquad
 \bAS((f))\le
       \frac{\sum_{i=1}^L\ord_{d_i}(q)}{\delta}+1.
 \label{app:mixed-order-bound}
\end{equation}
The corresponding estimates also hold for the highest-torsion
upper bound \((N-t(\C))/s_*+1\) of an arbitrary nonzero
simple-root cyclic ideal
when \(F_1,\ldots,F_L\) are the basic lifts of the irreducible
factors of the nonconstant generator of \(\Tor_{e-1}(\C)\).
\end{proposition}

\begin{proof}
Since \(\delta\mid\mu_i\), every \(F_i\) lies in \(A[X^\delta]\).
The same is true of \(f\), proving the divisibility.
Its degree is the sum of the displayed multiplicative orders,
and \eqref{app:mixed-order-bound} follows from
\eqref{app:free-as}.
For the last assertion, let \(\bar g\) be the highest-torsion
generator.  Its basic lift \(g\) has the same support period by
Theorem~\ref{coef:precision}, and
\(\deg g=N-t(\C)\).
Thus
\[
 \frac{N-t(\C)}{s_*}+1=\frac{\deg g}{\sigma_N(g)}+1.
 \]
The preceding argument, and the common-order argument when
applicable, may therefore be applied to \(g\).
If the highest torsion code is full, its generator is \(1\) and
the upper bound is simply \(1\).
\end{proof}

These estimates extend those of \cite{AS,SP} by applying coefficient
stability to the MED periods of the basic factors.  The exact product
period is given by Theorem~\ref{root:stabilizer} and may yield a
strictly smaller bound.

\begin{example}
\label{app:cancellation}
Over \(A=\Z_{3^e}\), let \(u\equiv1\pmod3\) be the Hensel lift
satisfying \(u^2=-2\).  The distinct basic-irreducible factors
\[
 F_+=X^2+uX-1,\qquad F_-=X^2-uX-1
 \]
have support period \(1\), but
\[
 F_+F_-=X^4+1
 \]
has support period \(4\).
Indeed, \(u\) is a unit, and the residue quadratics have
discriminant \(2\), a nonsquare in \(\F_3\).
At cyclic length \(8\), their product generates a free code with
arithmetic Singleton value \(4/4+1=2\).
Its core is the length-\(2\) code generated by \(Y+1\), which has
distance \(2\); hence the original code also has distance \(2\).
The common-order estimate only gives \(2\cdot2+1=5\).
The smaller exact bound reflects the larger stabilizer of the union
of the two root orbits.
\end{example}

\begin{example}
\label{app:unbounded}
Fix a prime \(\ell\ne p\), let \(b\ge2\), and put
\[
 N_b=\ell^b,\qquad
 f_b(X)=\frac{X^{N_b}-1}{X-1}
       =1+X+\cdots+X^{N_b-1}.
 \]
The polynomial \(f_b\) is a proper monic divisor, with
\(\spn(f_b)=1\), and it is reducible: its cyclotomic decomposition
contains the two positive-degree factors
\(\Phi_\ell\) and \(\Phi_{\ell^2}\).
The code \(\C_b=(f_b)\) consists exactly of the constant words
\((a,\ldots,a)\), \(a\in A\).  It is free of rank \(1\), and
every nonzero word has weight \(N_b\).  Consequently
\[
 d_H(\C_b)=\bAS(\C_b)=N_b=\ell^b\longrightarrow\infty.
 \]
Thus the arithmetic Singleton value is unbounded among codes with
reducible generators.  The bounds in
Theorem~\ref{app:common-order} and
Proposition~\ref{app:mixed-orders} depend on the root orders and
common support periods.
\end{example}

\section{Explicit families and computations}\label{sec:examples}

\subsection{A local lift over \texorpdfstring{\(\Z/16\Z\)}{Z/16Z}}

The factor \(G_1=X^3+X+1\) of \(X^7-1\) over \(\F_2\) has the
following successive Hensel lifts, obtained from
Theorem~\ref{coef:hensel}:
\begin{equation}\label{eq:lift-seven}
\begin{array}{c|l}
\text{coefficient ring}&\text{lift}\\ \hline
\Z/2\Z & X^3+X+1\\
\Z/4\Z & X^3+2X^2+X+3\\
\Z/8\Z & X^3+6X^2+5X+7\\
\Z/16\Z& X^3+6X^2+5X+15.
\end{array}
\end{equation}
For the first lift, monic division over \(\Z[X]\) gives remainder
\(2X^2-2\), so \(E=X^2+1\) over \(\F_2\). The complementary-factor
inverse is \(XG_1'=X^3+X\equiv1\pmod{G_1}\). The correction is
therefore \(E\), which gives the displayed factor modulo \(4\).
All four lifts have support period one. The coefficient of \(X^2\)
vanishes modulo \(2\) and first becomes nonzero modulo \(4\).
Thus Hensel lifting preserves the support period without preserving
the individual support positions.

\subsection{Powers of two over \texorpdfstring{\(\GR(3^e,1)\)}{GR(3e,1)}}

Let \(R=\Z/3^e\Z\), and let \(t\in R\) be the unique solution of
\[
 t^2=-2,\qquad t\equiv1\pmod3.
\]
Hensel's lemma gives existence and uniqueness because the derivative
is nonzero at \(t=1\) modulo \(3\).

\begin{proposition}\label{ex:powers-two}
For \(k\geq3\), the monic basic-irreducible factorization is
\begin{align}
 X^{2^k}-1={}&(X-1)(X+1)(X^2+1)\notag\\
 &{}\times\prod_{j=3}^k
 \bigl(X^{2^{j-2}}+tX^{2^{j-3}}-1\bigr)
 \bigl(X^{2^{j-2}}-tX^{2^{j-3}}-1\bigr).
 \label{eq:two-power-factorization}
\end{align}
Each of the last two factors for a fixed \(j\) has support period
\(2^{j-3}\) and support-primitive core \(Y^2\pm tY-1\).
\end{proposition}

\begin{proof}
For \(j\geq3\), one has
\(\ord_{2^j}(3)=2^{j-2}\), and the MED invariant is
\(\omega_{2^j}=2\). Hence the compression exponent of
Theorem~\ref{root:layer} is \(2^{j-3}\), and the reduced order is \(8\).
The two polynomials \(Y^2\pm tY-1\) reduce to irreducible quadratics
over \(\F_3\): their discriminant is \(2\), a nonsquare.
Their product is \(Y^4+1\), because \(t^2=-2\).
They are therefore the basic irreducible factors of \(\Phi_8\).
The factor correspondence in Theorem~\ref{root:layer} proves
irreducibility after the indicated substitutions and yields
\eqref{eq:two-power-factorization}. The linear and quadratic factors
at orders \(1,2,4\) are immediate. Their nonzero coefficients give the
claimed support periods.
\end{proof}

For example, \(t=22\) in \(\Z/81\Z\), and the two basic-irreducible
factors at order \(32\) are
\[
 X^8+22X^4-1,\qquad X^8+59X^4-1.
\]
Both have support period four and quadratic cores. Each associated
length-\(32\) cyclic code is an interleaving of four copies of a
length-eight code generated by the corresponding quadratic.
The arithmetic Singleton bound is \(3\), compared with the ordinary
Singleton bound \(9\).

The \(2k-1\) factors in \eqref{eq:two-power-factorization} have total
core degree \(4k-5\), counting repeated occurrences. The distinct
cores are \(Y-1\), \(Y+1\), and \(Y^2\pm tY-1\), whose degrees sum
to \(6\); in particular, \(X^2+1\) shares the core \(Y+1\) with
\(X+1\). Reusing these four cores reduces both the lifting cost and
the size of the stored coefficient arrays.

\subsection{Lengths dividing \texorpdfstring{\(q+1\)}{q+1}}

\begin{proposition}[Quadratic trace factors]
\label{v2:normone-factors}
Let \(A=\GR(p^e,r)\), let \(q=p^r\), and suppose that
\(n\mid q+1\). Choose a Teichm\"uller unit \(\theta\) of order
\(n\) in \(\GR(p^e,2r)\), and put
\[
 T_j=\theta^j+\theta^{-j}\in A,\qquad
 \varepsilon_n=\begin{cases}1,&2\mid n,\\0,&2\nmid n.\end{cases}
\]
Then the basic-irreducible factorization is
\begin{equation}\label{v2:eq-normone-product}
 X^n-1=(X-1)(X+1)^{\varepsilon_n}
       \prod_{j=1}^{\lfloor(n-1)/2\rfloor}(X^2-T_jX+1).
\end{equation}
All these factors are self-reciprocal. The trace coefficients satisfy
\[
 T_0=2,\qquad T_1=\theta+\theta^{-1},\qquad
 T_{j+1}=T_1T_j-T_{j-1}\quad(j\ge1).
\]
For \(1\le j<n/2\), one has
\[
 T_j=0\quad\Longleftrightarrow\quad
 4\mid n\ \hbox{ and }\ j=n/4;
\]
every other \(T_j\) in this range is a unit. Consequently the unique
quadratic factor with period two, when it occurs, is \(X^2+1\),
whose core is \(Y+1\). All other quadratic factors are
support-primitive and retain three nonzero coefficients at every
precision.
\end{proposition}

\begin{proof}
Since \(n\mid q+1\), one has \(p\nmid n\), and relative Frobenius
sends \(\theta^j\) to \(\theta^{qj}=\theta^{-j}\). Thus \(T_j\)
is fixed and belongs to \(A\). The fixed exponents under inversion
are \(0\) and, when \(n\) is even, \(n/2\). In the latter case
\(p\) is odd and \(\theta^{n/2}=-1\). All other exponents form
two-element Frobenius orbits \(\{j,-j\}\), giving the quadratic
factors in \eqref{v2:eq-normone-product}. Their reductions are the
corresponding irreducible orbit polynomials. Each root set is closed
under inversion, so all the factors are self-reciprocal. Expanding
\((\theta+\theta^{-1})(\theta^j+\theta^{-j})\) proves the recurrence.

If \(\bar T_j=0\), then \(\bar\theta^{2j}=-1\). For \(p=2\),
this would give \(n\mid2j\), contrary to \(0<2j<n\).
For odd \(p\), this equality forces \(n\) to be even, and uniqueness
of the Teichm\"uller lift gives \(\theta^{2j}=-1\), equivalently
\(2j\equiv n/2\pmod n\). In the stated range this holds precisely
when \(4\mid n\) and \(j=n/4\), and then \(T_j=0\).
In all other cases \(\bar T_j\ne0\), so \(T_j\) is a unit.
The support-period assertions follow from the displayed quadratics.
\end{proof}

\begin{corollary}[Hulls and self-duality when \(n\mid q+1\)]
\label{v2:normone-codes}
Under the hypotheses of Proposition~\ref{v2:normone-factors}, write
\(X^n-1=\prod_iF_i\), put \(A_i=A[X]/(F_i)\), and let
\(\mathcal C\) be a cyclic code with CRT description
\(\mathcal C\simeq\prod_i p^{a_i}A_i\), where \(0\le a_i\le e\).
Then
\begin{equation}\label{v2:eq-normone-hull}
 \mathcal C^\perp\simeq\prod_i p^{e-a_i}A_i,\qquad
 \Hull(\mathcal C)\simeq\prod_i p^{\max(a_i,e-a_i)}A_i,
\end{equation}
and
\[
 |\Hull(\mathcal C)|
   =q^{\sum_i\min(a_i,e-a_i)\deg F_i}.
\]
In particular, \(\mathcal C\) is LCD if and only if it is free.
There is a self-dual cyclic code of length \(n\) if and only if
\(e\) is even; when it exists, it is uniquely
\(p^{e/2}A^n\).
\end{corollary}

\begin{proof}
Let \(R=A[X]/(X^n-1)\), and let \(\iota\) be the involution
\(u(X)\mapsto u(X^{-1})\). For coefficient vectors \(u,v\), their
Euclidean inner product is the constant coefficient of
\(u\,\iota(v)\). Since a cyclic code is an ideal, testing
orthogonality against all shifts of its words gives
\[
 \mathcal C^\perp=\iota(\operatorname{Ann}_R(\mathcal C)).
\]
Indeed, all coefficients of \(u\,\iota(v)\) vanish if and only if
the constant coefficients of \(X^ju\,\iota(v)\) vanish for all
\(j\), and every \(X^ju\) remains in \(\mathcal C\).

Every \(F_i\) is self-reciprocal, so \(\iota\) preserves each CRT
component and each of its ideals \(p^aA_i\). The annihilator of
\(p^{a_i}A_i\) is \(p^{e-a_i}A_i\), proving the first formula in
\eqref{v2:eq-normone-hull}. Intersection of the two component ideals
gives the second, and their cardinalities give the hull size.
The hull is zero exactly when every \(a_i\) is \(0\) or \(e\),
which is precisely the freeness criterion in Proposition~\ref{pre:crt}.
Self-duality requires \(a_i=e-a_i\) for every \(i\). Thus \(e\)
must be even and every component must have depth \(e/2\), yielding
the unique code \(p^{e/2}R\), identified with \(p^{e/2}A^n\).
\end{proof}

\subsection{The order \texorpdfstring{\(p+1\)}{p+1}}

Let \(p\) be odd, \(n=p+1\), and \(R=\Z/p^e\Z\). Choose a primitive
Teichm\"uller \(n\)-th root \(\zeta\) in the unramified quadratic
extension of \(R\). Proposition~\ref{v2:normone-factors} gives
\begin{equation}\label{eq:p-plus-one}
 X^{p+1}-1=(X-1)(X+1)
 \prod_{j=1}^{(p-1)/2}(X^2-T_jX+1),
 \qquad T_j=\zeta^j+\zeta^{-j}\in R.
\end{equation}
The trace recurrence in that proposition is the Dickson recurrence
for this family \cite{WangDickson}.

\begin{proposition}\label{ex:p-plus-one}
The polynomial \(X^{p+1}-1\) splits completely over an unramified
extension of \(R\) if and only if the relative degree is even.
The least positive relative degree over which the basic-irreducible
factorization consists entirely of binomials is \(1\) for \(p=3\)
and \(2\) for \(p>3\).
\end{proposition}

\begin{proof}
Over \(\GR(p^e,t)\), relative Frobenius acts on exponents modulo
\(p+1\) by multiplication by \(p^t\). Since
\(p\equiv-1\pmod{p+1}\), this action is negation for odd \(t\) and
the identity for even \(t\). Thus an odd-degree extension leaves all
quadratic factors basic irreducible, while an even-degree extension
splits them into linear factors. There is at least one quadratic
factor because \(p\ge3\), so the first assertion follows.

For \(p=3\), the sole quadratic factor is \(X^2+1\), so every
basic-irreducible factor is already a binomial over \(R\).
For \(p>3\), there are \((p-1)/2\ge2\) quadratic factors. By
Proposition~\ref{v2:normone-factors}, at most one has zero linear
coefficient. At least one quadratic is therefore nonbinomial, and
it remains so in every odd-degree extension. A quadratic unramified
extension splits all factors into linear binomials, proving the
second assertion.
\end{proof}

For \(p\ge5\), pairing distinct quadratic factors with opposite
traces gives quartic generators with support period two or four.

\begin{proposition}\label{v2:opposite-trace-codes}
Let \(p\ge5\) be prime, let \(R=\Z/p^e\Z\), and put
\(n=p+1=2m\).  Retain the primitive root and traces in
\eqref{eq:p-plus-one}, and write
\[
 F_j(X)=X^2-T_jX+1\qquad(1\le j<m).
\]
For \(1\le j<m\) with \(2j\ne m\), put
\[
 P_j(X)=F_j(X)F_{m-j}(X),\qquad
 \C_j=(P_j)\subseteq R[X]/(X^n-1).
\]
Then \(\C_j\) is nonzero and free of rank \(n-4\), and
\begin{equation}\label{v2:eq-opposite-product}
 P_j(X)=X^4-T_{2j}X^2+1.
\end{equation}
Its support period is
\begin{equation}\label{v2:eq-opposite-period}
 \spn(P_j)=
 \begin{cases}
 4,&8\mid n\ \text{and }j\in\{n/8,3n/8\},\\
 2,&\text{otherwise}.
 \end{cases}
\end{equation}
Its minimum distance is
\begin{equation}\label{v2:eq-opposite-distance}
 d_H(\C_j)=
 \begin{cases}
 3,&m\text{ is odd and }\gcd(j,m)=1,\\
 2,&\text{otherwise}.
 \end{cases}
\end{equation}
In particular, \(\C_j\) attains its arithmetic Singleton bound if and
only if either \(m\) is odd and \(\gcd(j,m)=1\), or the first case of
\eqref{v2:eq-opposite-period} holds.  These are exactly the cases in
which its support-primitive core is MDS.
\end{proposition}

\begin{proof}
Since \(\zeta^m=-1\), one has
\[
 T_{m-j}=-T_j,\qquad T_j^2-2=T_{2j}.
\]
Multiplying the two quadratics proves
\eqref{v2:eq-opposite-product}.  The condition \(2j\ne m\) makes them
distinct, so \(P_j\) is a monic divisor of \(X^n-1\).
Its degree is \(4<n\), proving the rank and nonvanishing assertions.
The support period is four precisely when \(T_{2j}=0\).  This is
equivalent to
\[
 \zeta^{4j}=-1,\qquad\text{or equivalently}
 \qquad 4j\equiv n/2\pmod n.
\]
For \(1\le j<n/2\), the solutions are \(j=n/8\) and \(j=3n/8\),
and they exist exactly when \(8\mid n\).  Otherwise the middle
coefficient in \eqref{v2:eq-opposite-product} is nonzero and the
support period is two.

Set
\[
 H_j(Y)=Y^2-T_{2j}Y+1,
 \qquad \mathcal E_j=(H_j)\subseteq R[Y]/(Y^m-1).
\]
The roots \(\zeta^{2j},\zeta^{-2j}\) show that
\(H_j\mid Y^m-1\).  Since \(P_j(X)=H_j(X^2)\), the interleaver
identifies \(\C_j\) with \(\mathcal E_j^{\oplus2}\).
Thus \(d_H(\C_j)=d_H(\mathcal E_j)\), and
Proposition~\ref{app:highest-torsion} identifies this distance with
that of the residue code \(\overline{\mathcal E_j}\).

Let \(\alpha=\bar\zeta^{2j}\), of order
\(d=m/\gcd(j,m)\).  The restrictions on \(j\) give \(d>2\), and
\(\alpha^p=\alpha^{-1}\ne\alpha\); hence \(\bar H_j\) is its
irreducible quadratic polynomial over \(\F_p\).
The residue code has parameters \([m,m-2]\), contains no word of
weight one, and has distance at most three by the Singleton bound.
It contains a word of weight two exactly when
\[
 \alpha^h\in\F_p^\times\qquad\text{for some }1\le h<m.
\]
Indeed, a two-term word vanishing at \(\alpha\) gives this condition
by dividing its two terms.  Conversely, such an \(h\) gives the word
\(Y^h-\alpha^h\), of degree less than \(m\), in the residue code.
Since \(\alpha^p=\alpha^{-1}\), the condition is equivalent to
\(d\mid2h\).  Its least positive solution is
\(d/\gcd(d,2)\), which is less than \(m\) unless \(d=m\) is odd.
This proves \eqref{v2:eq-opposite-distance}.

Finally,
\[
 \bAS(\C_j)=\frac{4}{\spn(P_j)}+1.
\]
Comparison with \eqref{v2:eq-opposite-period} and
\eqref{v2:eq-opposite-distance} gives the equality cases.
Corollary~\ref{app:free-singleton} identifies them with the MDS cases
for the support-primitive core.
\end{proof}

\subsection{A free code attaining the arithmetic bound}

Let \(A=\GR(p^e,r)\), let \(p\nmid N\), and choose \(s\mid N\)
with \(M=N/s\geq2\). The cyclic code
\[
 \mathcal C=(X^s-1)\subseteq A[X]/(X^N-1)
\]
has support-primitive core
\(\mathcal D=(Y-1)\subseteq A[Y]/(Y^M-1)\).
The core consists of the words whose coordinate sum is zero.
It has rank \(M-1\), contains the word \((1,-1,0,\ldots,0)\), and
contains no word of weight one. Hence
\[
 d_H(\mathcal C)=d_H(\mathcal D)=2
       =\frac{\deg(X^s-1)}s+1.
\]
The arithmetic Singleton bound is attained for every \(e\) and \(r\).
The ordinary Singleton bound for \(\mathcal C\) is \(s+1\), so the
arithmetic bound is strictly sharper when \(s>1\).

\section{Concluding remarks}\label{sec:conclusion}

For a nonconstant monic divisor at a simple-root length, coefficient
compression, the multiplicative stabilizer of its root set, and the
interleaving of the associated constacyclic code determine the same
support period. Hensel uniqueness preserves this period at every coefficient
precision. For cyclic Frobenius-orbit polynomials, the MED invariant
gives the support period and core degree. For arbitrary simple-root
ideals, the intrinsic period is the greatest common divisor of the
periods of the torsion codes.

The coordinate-projection decomposition remains valid at arbitrary
lengths, and the highest torsion code determines the minimum Hamming
distance of every nonzero code over \(A\). At repeated-root lengths,
nonzero free codes obtained by a
\(p^a\)-power substitution from simple-root generators preserve
the support-primitive core, distance, and arithmetic Singleton value.
For general repeated-root ideals, even the complete torsion
filtration need not determine the intrinsic period.

\appendix
\renewcommand{\thetheorem}{\Alph{section}.\arabic{theorem}}
\section{Construction and bit complexity}\label{sec:complexity}

We estimate the bit cost of constructing the unordered
basic-irreducible factorization of \(X^n-1\), where \(p\nmid n\).
The output representations are defined below, and the arithmetic
bounds use standard polynomial algorithms
\cite{ModernComputerAlgebra,CantorKaltofen}. Integer preprocessing,
residue-field factorization, ring lifting, and output are charged
separately. Matching factors to prescribed root labels is an
additional task.

\subsection{Representation and deterministic lifting}

Fix a supplied basic irreducible polynomial
\(\Pi(Z)\in(\Z/p^e\Z)[Z]\) of degree \(r\), and represent
\[
 A_s=(\Z/p^s\Z)[Z]/(\Pi\bmod p^s),\qquad
 \rho_s=\lceil s\log_2p\rceil.
\]
Each element occupies \(\Theta(r\rho_s)\) bits. The construction of
\(\Pi\) is separate from the costs below.
Let \(M_{\Z}(b)\geq b\) be a regular integer multiplication bound,
with \(M_{\Z}(b)/b\) nondecreasing and
\(M_{\Z}(2b)=O(M_{\Z}(b))\). We write \(\widetilde O\) to suppress
polylogarithmic factors in the indicated size parameters.
Let \(M_A(d,s)\) bound polynomial multiplication, modular
multiplication, and monic division at degrees \(O(d)\) over \(A_s\),
including any necessary precomputation.
We choose these bounds monotone in \(d,s\) and superadditive in \(d\).

\begin{lemma}\label{alg:arithmetic}
One may take
\[
 M_A(d,s)=O\bigl(d^2r^2[M_{\Z}(\rho_s)+\rho_s]\bigr)
\]
with schoolbook polynomial arithmetic, and
\[
 M_A(d,s)=\widetilde O(dr\rho_s)
\]
with fast arithmetic. Addition, differentiation, multiplication by a
prime-subring scalar, and exact coefficientwise division by a known
power of \(p\) cost
\(O(dr[M_{\Z}(\rho_s)+\rho_s])\).
For either bound, the total cost over a precision-doubling schedule
has the same order as the cost at the final precision.
\end{lemma}

\begin{proof}
A coefficient product followed by reduction modulo \(\Pi\) takes
\(O(r^2)\) operations on \(\rho_s\)-bit integers by schoolbook
arithmetic, proving the first bound. For the fast bound, encode the
two-variable coefficient array by Kronecker substitution, allowing
\(O(\rho_s+\log(dr+1))\) bits per slot for coefficient growth.
Fast integer multiplication and fast monic reduction give the stated
soft bound. For fast division, the reversed monic divisor \(J(T)\) has constant
term one. Its reciprocal is refined by
\[
 V_{2u}=V_u(2-JV_u)\pmod {T^{2u}},
 \qquad 1-JV_{2u}=(1-JV_u)^2\pmod {T^{2u}},
\]
using only ring operations. The remaining operations act on \(dr\)
stored coordinates. Differentiation uses successive additions to
maintain its integer multipliers modulo \(p^s\). The target precisions
sum to \(O(e)\), and regularity bounds the resulting geometric sum by
a constant multiple of the final-precision bound.
\end{proof}

For each \(d\mid n\), let \(m_d=\ord_d(q)\), let \(\omega_d\) be the
MED invariant of Proposition~\ref{pre:med}, and set
\[
 \mu_d=m_d/\omega_d,\qquad D_d=d/\mu_d,\qquad
 \mathscr D=\{D_d:d\mid n\}.
\]
Use \(m_1=\omega_1=\mu_1=D_1=1\). The construction from
Theorem~\ref{root:layer} has three stages.

\begin{enumerate}
\item Compute the arithmetic data \(d,m_d,\omega_d,\mu_d,D_d\), and
group the divisors \(d\) by their common value \(D_d\).
\item For each distinct \(D\), factor \(\Phi_D\) over \(\F_q\).
Lift its factors using Theorem~\ref{coef:hensel} with ambient
polynomial \(Y^D-1\).
\item For each original layer \(d\), output \(H(X^{\mu_d})\) for
every stored factor \(H\) of \(\Phi_{D_d}\).
\end{enumerate}

Let \(\mathfrak A(n,q)\) denote the cost of a deterministic implementation
of the first step, including integer factorization, divisor construction,
order computations, and
the modular powers used to represent \(q=p^r\). It excludes both
construction of \(\Pi\) and finite-field polynomial factorization.
Put
\begin{align}
 s_D&=\varphi(D),& w_D&=\ord_D(q),&
 L_n&=\lceil\log_2(n+1)\rceil,\notag\\
 W&=\sum_{d\mid n}\varphi(D_d),&
 U&=\sum_{D\in\mathscr D}s_D,&
 V&=\sum_{D\in\mathscr D}s_Dw_D.
 \label{eq:compressed-sizes}
\end{align}
The conventions \(\varphi(1)=w_1=1\) apply. The total degree of
the cores, with repetition, is \(W\), whereas \(U\) is the total
degree to be lifted after shared cores are identified. In particular,
\[
 U\leq W,\qquad V\leq U^2\leq W^2.
\]
Indeed, the primitive \(d\)-th layer has
\(\varphi(d)/m_d=\varphi(D_d)/\omega_d\) factors of core degree
\(\omega_d=w_{D_d}\).

In a fixed-stride representation, all nonleading coefficients of the
core \(H\), including zeros, are stored together with its stride and
degree. If \(h\) is the total number of factors, this representation
has size
\begin{equation}\label{eq:stride-size}
 S_{\mathrm{stride}}=\Theta(Wr\rho_e+hL_n).
\end{equation}
An explicit coefficient--exponent list uses
\(O(W(r\rho_e+L_n))\) bits. Repeated arrays may be shared, and additional zero coefficients may
be omitted. The stated sizes charge each occurrence separately and
include all coefficients of each core, irrespective of further
cancellations.

\begin{theorem}\label{alg:compressed}
Suppose that the arithmetic data and the factorizations of
\(\Phi_D\bmod p\) over \(\F_q\), for \(D\in\mathscr D\), are supplied.
The ring lifting and fixed-stride output can be constructed
deterministically in
\begin{equation}\label{eq:fast-compressed}
 \widetilde O(Ure\log p)+O(S_{\mathrm{stride}})
 =\widetilde O(Wre\log p+WL_n)
\end{equation}
bit operations. In the coefficient-computation term, the suppressed logarithmic
factors depend only on the compressed sizes and the ring parameters.
\end{theorem}

\begin{proof}
For \(D\geq2\), binary modular powering of \(Y^D\) uses
\[
 \eta(D)=\lfloor\log_2D\rfloor+\operatorname{pop}_2(D)-1
\]
products, where \(\operatorname{pop}_2(D)\) counts binary ones.
Set \(\eta(1)=0\). There are \(s_D/w_D\) factors of degree \(w_D\).
After this powering, the local correction uses a bounded number of
modular products and a linear number of coefficient operations. Lemma~\ref{alg:arithmetic},
summed over the lifting precisions, therefore gives
\(\widetilde O((\eta(D)+1)s_Dr\rho_e)\) for this layer.
The scalar inverse \(D^{-1}\bmod p^e\) is computed once and shared,
at cost \(\widetilde O(\log(D+1)+\rho_e)\).

The inequality
\begin{equation}\label{eq:totient-size}
 D\leq2\varphi(D)^2=2s_D^2
\end{equation}
absorbs \(\eta(D)\) into logarithmic factors in the compressed size.
To prove it, multiply the ratios
\(\varphi(\ell^a)^2/\ell^a=\ell^{a-2}(\ell-1)^2\)
over the prime powers dividing \(D\). Every ratio is at least one,
except the ratio \(1/2\) for \(2^1\).
Summing the layer costs gives the first term in
\eqref{eq:fast-compressed}. Emitting the repeated core arrays and
stride fields costs \(O(S_{\mathrm{stride}})\). Since
\(h\leq W\) and \(U\leq W\), the displayed bound follows.
\end{proof}

For schoolbook arithmetic, put
\[
 K=\sum_{D\in\mathscr D}(\eta(D)+1)s_Dw_D.
\]
The same proof gives the explicit schoolbook bound
\begin{align*}
 O\bigl(&r^2K[M_{\Z}(\rho_e)+\rho_e]
       +Ur[M_{\Z}(\rho_e)+\rho_e]+S_{\mathrm{stride}}\bigr)\\
 &\quad+\sum_{\substack{D\in\mathscr D\\D>1}}
            \widetilde O(\log(D+1)+\rho_e).
\end{align*}
With \(M_{\Z}(b)=O(b^2)\), the dependence on the coefficient
precision is quadratic.

\subsection{Finite-field factorization and total cost}

\begin{theorem}\label{alg:complete}
Given the coefficient-ring presentation, the complete unordered
basic-irreducible factorization of \(X^n-1\) can be constructed by a
Las Vegas algorithm with expected bit cost
\begin{equation}\label{eq:full-cost}
 \mathfrak A(n,q)+
 \widetilde O\bigl(V(r\log p)^2+Wre\log p+WL_n\bigr).
\end{equation}
The output may use either fixed strides or explicit
coefficient--exponent pairs. Randomization is confined to the
finite-field factorization.
\end{theorem}

\begin{proof}
For \(D=1\), the residue factor is \(Y-1\). For \(D>1\), construct
\(\Phi_D\bmod p\) using the M\"obius product
\[
 \Phi_D(Y)=\prod_{u\mid D}(Y^{D/u}-1)^{\mu_{\mathrm M}(u)},
\]
placing positive and negative exponents in a numerator and a
denominator, respectively, and performing one exact monic division.
If \(\ell_1<\cdots<\ell_k\) are the prime divisors of \(D\), the sum
of all leaf degrees is
\[
 D\prod_{i=1}^k(1+\ell_i^{-1})
 \leq \varphi(D)\frac{(k+1)(k+2)}2.
\]
Indeed, \(\ell_i\geq i+1\) gives
\(\prod_i\ell_i/(\ell_i-1)\leq k+1\) and
\(\prod_i(1+\ell_i^{-1})\leq(k+2)/2\).
Product trees and monic division therefore construct this polynomial
in \(\widetilde O(s_D(b+\log(D+1)))\) bit operations, where
\(b=r\lceil\log_2p\rceil\).

All irreducible factors of \(\Phi_D\) over \(\F_q\) have degree
\(w=w_D\). Write \(s=s_D\). Use equal-degree splitting
\cite{CantorZassenhaus}, maintaining blocks whose product is
\(\Phi_D\). For a block \(B\) of degree greater than \(w\), choose
a uniform residue \(a\) modulo \(B\). For odd \(p\), split with
\[
 \gcd\bigl(a^{(q^w-1)/2}-1,B\bigr).
\]
For \(p=2\), use the absolute trace and split with
\[
 \gcd\left(\sum_{i=0}^{rw-1}a^{2^i},B\right).
\]
In each pass, apply this step once to every block of degree greater
than \(w\), retaining any block for which the split is trivial.
For two distinct irreducible factors in the same block, the random
components are independent by the Chinese remainder theorem.
For odd \(p\), each lies in the selected half with probability
\(\theta=(q^w-1)/(2q^w)\), and
\(\theta^2+(1-\theta)^2\leq5/9\). For \(p=2\), the nonzero absolute
trace functional gives nonseparation probability \(1/2\).
After \(j\) passes, the probability that any pair remains in the same
block is at most \(\binom{s/w}{2}(5/9)^j\). Summing this tail gives
an expected \(O(1+\log(s/w))\) passes.

Powering or trace evaluation uses \(O(w\log q)=O(wb)\) modular
products per block. Fast arithmetic, monic gcds, and uniform sampling
therefore cost \(\widetilde O(swb^2)\) per layer in expectation.
The total degree of the active blocks is at most \(s\) in each pass.
At termination, every block has degree \(w\) and is therefore
irreducible. The preceding tail bound proves almost-sure termination.
Summing over \(D\) gives \(\widetilde O(Vb^2+UL_n)\), including
construction of the cyclotomic polynomials. Combine this with
Theorem~\ref{alg:compressed}. Explicit exponents are emitted by
successive additions of the stride at cost \(O(L_n)\) per position,
which is included in \(WL_n\).
\end{proof}

\subsection{Certification and simultaneous lifting}
\label{v2:certification-batching}

The compressed factorization admits a certificate on the stored core
polynomials. The arithmetic data and the residue factorizations are
inputs to this certificate.

\begin{proposition}[Certificate for compressed factors]
\label{v2:factorization-certificate}
Let \(p\nmid n\), and suppose that the data
\(\mu_d,D_d\), for \(d\mid n\), satisfy
Theorem~\ref{root:layer}. For each \(D\in\mathscr D\), let
\(\mathcal G_D\) be a certified complete set of distinct monic
irreducible factors of \(\Phi_D\bmod p\) over \(\F_q\).
For every \(g\in\mathcal G_D\), suppose that a monic polynomial
\(H_{D,g}\in A[Y]\) satisfies
\begin{equation}\label{v2:core-certificate}
 \overline{H_{D,g}}=g,\qquad
 \deg H_{D,g}=\deg g,\qquad
 \remn_{H_{D,g}}(Y^D-1)=0.
\end{equation}
Then
\[
 \Phi_D(Y)=\prod_{g\in\mathcal G_D}H_{D,g}(Y)
\]
for every \(D\in\mathscr D\), and
\begin{equation}\label{v2:certified-output}
 X^n-1=
 \prod_{d\mid n}\ \prod_{g\in\mathcal G_{D_d}}
       H_{D_d,g}(X^{\mu_d})
\end{equation}
is the complete factorization into distinct monic basic-irreducible
polynomials. Thus the full factorization can be certified without
expanding the substitutions in \eqref{v2:certified-output}.
\end{proposition}

\begin{proof}
Fix \(D\in\mathscr D\), and write
\(Y^D-1=\Phi_D(Y)K_D(Y)\) in \(A[Y]\). Since \(p\nmid D\),
the reduction of \(Y^D-1\) is squarefree. Therefore
\(\gcd(g,\overline{K_D})=1\) for every \(g\in\mathcal G_D\),
so \(H_{D,g}\) and \(K_D\) are comaximal. Reducing the product
identity modulo \(H_{D,g}\), condition~\eqref{v2:core-certificate}
and the invertibility of \(K_D\) in \(A[Y]/(H_{D,g})\) give
\(H_{D,g}\mid\Phi_D\).

Distinct \(H_{D,g}\) are pairwise comaximal because their reductions
are distinct irreducibles. Hence their product divides \(\Phi_D\):
for pairwise comaximal ideals, intersection equals product.
The product and \(\Phi_D\) are monic and have the same degree
\(\sum_g\deg g=\varphi(D)\), so they are equal. Each factor is
basic irreducible by its reduction. Theorem~\ref{root:layer} now
identifies the factors in each primitive-order \(d\) layer, and
\(\prod_{d\mid n}\Phi_d=X^n-1\) proves
\eqref{v2:certified-output}. Distinct layers have distinct root
orders, and the factors within each layer are distinct.
\end{proof}

The divisibility tests in \eqref{v2:core-certificate} use modular
powering at the core degrees. By \eqref{eq:totient-size}, their
coefficient cost is covered by the
\(\widetilde O(Ure\log p)\) term in
Theorem~\ref{alg:compressed}; reading the chosen output representation
has its stated output cost. Certification of the integer data and
residue factorizations remains part of the preprocessing.

For simultaneous lifting, the local corrections are the coordinates
of a single global error under a Chinese remainder isomorphism.
This description holds for every monic polynomial with squarefree
reduction.

\begin{proposition}[Local corrections and the global error]
\label{v2:crt-correction}
Let \(f\in A[X]\) be monic of degree \(n\ge1\), with squarefree
reduction. Let \(1\le t\le s\) and \(s+t\le e\), and suppose
that \(h\ge1\) monic polynomials of positive degrees
\(d_1,\ldots,d_h\) form a factorization of \(f\) over \(A_s\).
Choose monic lifts \(G_i\in A_{s+t}[X]\) of the same degrees and put
\[
 P=\prod_{i=1}^hG_i,\qquad
 E=\frac{f-P}{p^s}\in A_t[X],\qquad
 g_i=G_i\bmod p^t,\qquad Q_i=\prod_{j\ne i}g_j.
\]
Then \(\deg E<n\). There are unique polynomials
\(\Delta_i\in A_t[X]\), with \(\deg\Delta_i<d_i\), such that
\begin{equation}\label{v2:crt-error-identity}
 E=\sum_{i=1}^h\Delta_iQ_i.
\end{equation}
They are given by
\[
 \Delta_i=\remn_{g_i}(E Q_i^{-1}),
\]
where \(Q_i^{-1}\) is taken modulo \(g_i\), and they satisfy
\begin{equation}\label{v2:simultaneous-product}
 f=\prod_{i=1}^h(G_i+p^s\Delta_i)
       \qquad\text{in }A_{s+t}[X].
\end{equation}
The updated factors are independent of the chosen monic lifts \(G_i\).
If \(f=X^n-\lambda\), with \(p\nmid n\) and
\(\lambda\in A^\times\), the inverse is the one in
\eqref{eq:inverse}:
\[
 Q_i^{-1}\equiv(n\lambda)^{-1}Xg_i'\pmod{g_i}.
\]
\end{proposition}

\begin{proof}
The leading terms of \(f\) and \(P\) cancel, and
\(f-P\in p^sA_{s+t}[X]\). Thus exact coefficientwise division
by \(p^s\) defines \(E\in A_t[X]\) with \(\deg E<n\).
The polynomials \(g_i\) are pairwise comaximal, because their
reductions form a factorization of the squarefree polynomial
\(\bar f\). In particular, \(Q_i\) is invertible modulo \(g_i\).
The displayed formula for \(\Delta_i\) gives
\[
 E-\sum_i\Delta_iQ_i\equiv0\pmod{g_j}
 \qquad(1\le j\le h).
\]
Pairwise comaximality makes this difference divisible by
\(\prod_i g_i\). Its degree is below \(n\), so it is zero.
The same congruences prove uniqueness. Expanding the product in
\eqref{v2:simultaneous-product}, all terms containing two corrections
vanish because \(2s\ge s+t\); the remaining terms are
\(P+p^s\sum_i\Delta_iQ_i=P+p^sE=f\).

To check independence of the chosen lifts, replace \(G_i\) by
\(G_i+p^sU_i\), where \(\deg U_i<d_i\) and
\(U_i\in A_t[X]\). The new global error is
\(E-\sum_iU_iQ_i\), so uniqueness in
\eqref{v2:crt-error-identity} replaces \(\Delta_i\) by
\(\Delta_i-U_i\). The updated factors are consequently unchanged.
For \(f=X^n-\lambda\), the factorization
\(f=\prod_i g_i\) in \(A_t[X]\) gives, after differentiation,
\(Q_i g_i'\equiv nX^{n-1}\pmod{g_i}\). Multiplication by \(X\)
and the relation \(X^n\equiv\lambda\pmod{g_i}\) yield the inverse
formula.
\end{proof}

\begin{proposition}[Batched extraction of local errors]
\label{v2:batched-errors}
Use the notation of Proposition~\ref{v2:crt-correction}, and put
\(\ell_h=\lceil\log_2(h+1)\rceil\). A product tree at precision
\(p^{s+t}\), followed by a remainder tree at precision \(p^t\),
computes all local errors
\begin{equation}\label{v2:batched-local-errors}
 E_i=\frac{\remn_{G_i}(f)}{p^s}
     =\remn_{g_i}(E)\in A_t[X].
\end{equation}
Their total bit cost is
\begin{equation}\label{v2:batched-error-cost}
 O\!\left(
   \ell_h M_A(n,s+t)
   +nr\bigl[M_{\Z}(\rho_{s+t})+\rho_{s+t}\bigr]
 \right).
\end{equation}
If \(f=X^n-\lambda\), with \(p\nmid n\) and
\(\lambda\in A^\times\), and the common inverse
\((n\lambda)^{-1}\) is supplied, all corrected factors are obtained
within the same bound. They agree with the factorwise lifts from
Theorem~\ref{coef:hensel}.
\end{proposition}

\begin{proof}
The product \(P\) is divisible by each \(G_i\), so
\(\remn_{G_i}(f)=\remn_{G_i}(p^sE)\).
Monic division commutes with reduction, and its remainder is linear
in the dividend. Consequently exact division by \(p^s\) gives
\eqref{v2:batched-local-errors}. In particular, once \(E\) is
computed, the remainder tree uses only precision \(p^t\).

Construct a balanced binary product tree with leaves \(G_i\).
Its height is at most \(\ell_h\), and the total degree at each
level is at most \(n\). Superadditivity of \(M_A\) bounds the
multiplications at one level by \(O(M_A(n,s+t))\). This proves
the product-tree cost \(O(\ell_hM_A(n,s+t))\).
Subtracting \(P\) from \(f\) and dividing its nonleading
coefficients by \(p^s\) has the second cost in
\eqref{v2:batched-error-cost}.

Reduce the node products modulo \(p^t\). Starting with \(E\) at
the root, reduce each incoming remainder modulo the one or two child
products. At a node of degree \(d\), the incoming remainder has
degree below \(d\), so the child reductions cost
\(O(M_A(d,t))\). Summing over one level gives
\(O(M_A(n,t))\), again by superadditivity. Thus the remainder tree
costs \(O(\ell_hM_A(n,t))\), which is absorbed in
\eqref{v2:batched-error-cost} by monotonicity in the precision.
The leaves contain the \(E_i\).

For the binomial ambient polynomial, the inverse formula in
Proposition~\ref{v2:crt-correction} gives the local corrections using
\(O(M_A(d_i,t))\) bit operations at each leaf, together with the linear
coefficient operations of Lemma~\ref{alg:arithmetic}. Their total
cost is absorbed by \eqref{v2:batched-error-cost}, since
\(\sum_iM_A(d_i,t)\le M_A(n,t)\). Multiplication by the general
\(A_t\)-scalar \((n\lambda)^{-1}\) is included in these polynomial
multiplication costs. The CRT identity proves the
simultaneous product relation, and Theorem~\ref{coef:hensel}
identifies the updated factors with the factorwise lifts.
\end{proof}

\subsection{Dense output, orbit labels, and primary blocks}

For dense output, let \(m_C\) be the degrees of the basic factors and
put
\[
 \Delta_2=\sum_Cm_C^2,\qquad
 n\leq\Delta_2\leq n\max_Cm_C\leq n^2.
\]
If the residue factorization is supplied, factorwise application of
Theorem~\ref{coef:hensel} has fast cost
\(\widetilde O((\eta(n)+1)nre\log p)\), and schoolbook cost
\[
 O\bigl((\eta(n)+1)r^2\Delta_2[M_{\Z}(\rho_e)+\rho_e]
          +nr[M_{\Z}(\rho_e)+\rho_e]\bigr),
\]
apart from the shared scalar inverse. The exact root order can
replace \(n\) in each factor's powering cost by
Proposition~\ref{coef:enlargement}.

Alternatively, Proposition~\ref{v2:batched-errors} computes all
factor remainders simultaneously. Summing its bound over the
precision-doubling schedule gives fast cost
\(\widetilde O(nre\log p)\); the schoolbook estimate replaces
\((\eta(n)+1)\Delta_2\) by
\(\lceil\log_2(h+1)\rceil n^2\).
The dense coefficient arrays have exactly \(n\) nonleading slots.
Thus, given the residue factors, batched lifting is quasi-linear in
the dense output size, as in standard fast Hensel lifting
\cite{ModernComputerAlgebra}.

For prescribed cyclotomic labels, the leader formula in
Proposition~\ref{pre:med} computes the least element of
\(c_{n/q}(\gamma)\) from \(\omega_C\) modular iterates modulo
\(\omega_Cn/m_C\). Given the representatives and the required
residues of \(q\), all leaders cost
\[
 O\bigl(WM_{\Z}(L_n)\bigr).
\]
If increasing primitive leader lists at every order \(d\mid n\)
are supplied, scaling by \(n/d\) and merging the lists costs
\[
 O\bigl(h[M_{\Z}(L_n)+L_n\log(\tau(n)+1)]\bigr),
\]
where \(\tau(n)\) is the number of divisors of \(n\).
The merge is needed only for globally ordered output. Constructing
the representative lists and matching unordered factors to specified
roots have separate costs. The required modular powers of \(q=p^r\)
are included in \(\mathfrak A(n,q)\).

For \(N=p^am\) over \(\Z/p^e\Z\), with \(p\nmid m\),
Proposition~\ref{coef:repeated-blocks} constructs the canonical local
blocks by multiplying every simple-root stride by \(p^a\).
Consequently their expected construction cost is
\[
 \mathfrak A(m,p)+
 \widetilde O\bigl(V_m(\log p)^2+W_me\log p+
                              W_m\log(N+1)\bigr).
\]

For the family in Proposition~\ref{ex:powers-two},
\(\mathscr D=\{1,2,8\}\), \(U=6\), \(W=4k-5\), and \(h=2k-1\).
Every displayed coefficient is nonzero. With fixed-width coefficients
and ordinary binary exponents, the explicit sparse output has size
\(\Theta(ke+k^2)\) bits: the exponent lengths summed over the layers
are \(\Theta(k^2)\). Lifting the two fixed quadratics and emitting
the output takes \(\widetilde O(ke+k^2)\) bit operations, while dense
coefficient output requires \(\Omega(2^ke)\) bits.

\section*{Acknowledgements}

The authors acknowledge the assistance of GPT in refining the language
and presentation of the manuscript.

\end{document}